\documentclass[a4paper,fleqn]{cas-dc}

\usepackage[utf8]{inputenc}
\usepackage{amsthm}
\usepackage{subcaption}
\usepackage{setspace}
\usepackage{mathtools,cuted}
\usepackage{amssymb}
\usepackage{graphicx}
\usepackage{multirow}
\usepackage{bigstrut}
\usepackage{hyperref}
\usepackage{orcidlink}
\usepackage{natbib}

\newtheorem{theorem}{Theorem}
\newtheorem{corollary}{Corollary}
\newtheorem{lemma}[theorem]{Lemma}

\newtheorem{assumption}{Assumption}
\newtheorem{proposition}{Proposition}
\newtheorem{remark}{Remark}
\newtheorem{problem}{Problem statement}
\mathtoolsset{showonlyrefs,showmanualtags}

\begin{document}
	\let\WriteBookmarks\relax
	\def\floatpagepagefraction{1}
	\def\textpagefraction{.001}
	\shorttitle{Finite-Time Trajectory Tracking of a FWMR}
	\shortauthors{B Anil et~al}
	
	\title [mode = title]{Finite-Time Trajectory Tracking of a Four Wheeled Mecanum Robot}                      

\author[1]{Anil B}[orcid=0000-0001-5232-1314]
\cormark[1]
\ead{122004002@smail.iitpkd.ac.in}
%\credit{Conceptualization, Formal Analysis, Methodology, Software,  Investigation, Writing-original draft, Visualization}

\author[1]{Mayank Pandey}
\fnmark[1]
%\credit{Conceptualization, Formal Analysis, Methodology, Software}

\author[1]{Sneha Gajbhiye}
\fnmark[2]
%\credit{Project administration, Writing-review \& editing, Supervision, Funding acquisition}
	
\affiliation[1]{organization={Department of Electrical Engineering, Indian Institute of Technology Palakkad},
		addressline={Kanjikode}, 
		city={Palakkad},
		postcode={678623}, 
		state={Kerala},
		country={India}}

\cortext[cor1]{Corresponding author.}
\fntext[fn1]{\textit{E-mail address:} mayankpandey20001@gmail.com (M. Pandey). }
\fntext[fn2]{\textit{E-mail address:} snehagajbhiye@iitpkd.ac.in (S. Gajbhiye). }

	\begin{abstract}
	Four Wheeled Mecanum Robot (FWMR) possess the capability to move in any direction on a plane making it a cornerstone system in modern industrial operations. Despite the extreme maneuverability offered by FWMR, the practical implementation or real-time simulation of Mecanum wheel robots encounters substantial challenges in trajectory tracking control. In this research work, we present a finite-time control law using backstepping technique to perform stabilization and trajectory tracking objectives for a FWMR system. A rigorous stability proof is presented and explicit computation of the finite-time is provided.  For tracking objective, we demonstrate the results taking an S-shaped trajectory inclined towards collision avoidance applications. Simulation validation in real time using Gazebo-ROS on a Mecanum robot model is carried out which complies with the theoretical results.
	\end{abstract}
	
	%\begin{graphicalabstract}
	%\includegraphics{figs/cas-grabs.pdf}
	%\end{graphicalabstract}

	\begin{keywords}
		Nonlinear control \sep Mobile Robots \sep Backstepping 
	\end{keywords}
	\maketitle	
	\section{Introduction}
	
	Four Wheeled Mecanum Robots (FWMRs) exhibit the ability to navigate freely in all directions on a flat surface. This omnidirectional mobility feature makes them widely employed in various contexts, including industrial operations, constrained environments like warehouses, inspection zones, search and rescue missions, cleaning duties, and assistive robots \citep{zhang2023application}. Hence, across various sectors such as manufacturing, logistics, and unmanned warehousing, the mecanum wheel robot is renowned for its remarkable ability to execute zero-radius turns effortlessly in any direction as well as its adeptness in navigating within confined spaces \citep{taheri2020omnidirectional}. Its widespread adoption is a testament to its versatility and efficiency, making it a cornerstone technology in modern industrial operations. Consequently, FWMR's have garnered significant attention in research circles. To enable the mobile robot to follow the desired trajectory accurately, it is imperative to develop precise tracking controllers. A model predictive control strategy equipped with compensation for friction is designed in \citep{ref1} to facilitate tracking objective. Output-feedback strategy tailored with continuous sliding mode controllers is employed in \citep{ref2} for  addressing the robust trajectory tracking challenges encountered by FWMR. A sliding mode controller derived through backstepping technique tailored for robust control of FWMR  is presented in \citep{ref3}. A variant of computed-torque control approach is  developed in \citep{ref4} for addressing  the tracking control objective of FWMR. Researchers have developed a differential sliding mode controller aimed at facilitating the movement of an omnidirectional mobile platform in \citep{ref5}. A robust adaptive terminal sliding mode control scheme is designed in \citep{ref6} specifically to address the tracking challenges encountered by an FWMR targeted for examination of aircraft skin. Regarding transient performance, which aims to limit the tracking error within a specified range, a widely adopted strategy involves employing an error transformation function. This technique is illustrated in \citep{new19} to ensure the transient performance of underactuated surface vessels for executing trajectory tracking. 
	%Similarly, in [20], performance functions are formulated, and the error transformation function is applied to ensure that the tracking error of  hypersonic vehicles remains within predetermined boundaries. 
	In \citep{new21}, the concept of an error transformation function is presented, which effectively converts the constrained tracking objective into an unconstrained stabilization problem for marine surface vessel. Additionally, in \citep{new22}, a performance function based on finite-time is implemented to converge the tracking errors within a specified range in the context of tracking problems involving systems with uncertain dynamics. The conventional asymptotic stability theory can solely ensure the asymptotic stability of a control object concerning finite-time convergence. The emergence of finite-time stability theory aims to achieve stabilization within a finite duration, ensuring that the tracking error converges within a bounded interval. Initially proposed by \citep{ref23}, subsequent variants of finite-time stability theory \citep{ref25,ref26,ref27,ref28,bohn2016almost} have been introduced for enhanced practical application. In \citep{ref29}, a finite-time adaptive neural controller, with optimized design parameters, is developed to address the finite-time optimal control challenges for a specific class of nonlinear systems. A finite-time tracking objective is addressed for a class of  nonlinear Multi-Input and Multi-Output (MIMO) systems in \citep{ref30} using neural networks in an adaptive setting, whose performance depends on the number of training samples considered. Despite the inherent advantages of extreme manueurability of FWMR, the practical implementation or realistic simulation of the Mecanum wheel robot faces significant hurdles in control, a critical aspect marked by numerous challenges. Factors such as nonlinearity, external perturbations, uncertainty etc.,  renders conventional control techniques inadequate for meeting stringent trajectory tracking demands. Furthermore, issues like significant drift from intended trajectories, oscillatory nature of control signals and time delays may arise as a consequence. Conventional control schemes often face difficulties to effectively tackle tracking control issues under these complex practical conditions resulting in inefficiencies, inaccuracies, and unreliability in handling control tasks \citep{tang2024mecanum}. Based on an elaborative literature survey, there is not much
	adequate research implemented in the validation of finite-time trajectory tracking controllers in a realistic 
	simulation environment for FWMR system. Towards this effort, our contributions  are as follows:
	\begin{enumerate}
		\item We present a finite-time backstepping tracking control to track a desired non-trivial trajectory. The proposed control law is finite-time continuous and gives finite-time global convergence of states to the equilibrium. 
		\item A closed loop stability proof in finite-time setting with an explicit computation of finite-time is provided.
		\item An estimation of upper norm bound of disturbance torque that can be tolerated with the proposed controller is provided and a case study associated with the same is demonstrated in the simulation analysis.
		\item The proposed control algorithm is validated in rGazebo ROS realistic simulation platform inorder to demonstrate the efficiency of the controller under such practical conditions as outlined above. 
		\item We also present a comparison of finite-time and asymptotic control results for the FWMR carried out in Gazebo-ROS.  
		%	The results demonstrate faster convergence of system states to desired values for finite-time case in contrast with the asymptotic case.
	\end{enumerate}
	Organization of the paper is as follows: Section \ref{sec:2} and  \ref{sec:3} discusses the mathematical preliminaries and system model respectively. The theoretical developments are presented in section \ref{sec:4}  followed by simulation results in \ref{sec:5} and conclusions in section \ref{sect:6}.
	
	\section{Mathematical preliminaries} \label{sec:2}
	In this section, we present the key elements necessary for the theoretical developments outlined in this paper.  \par 
	\textit{Notations:} $\mathbb{R}^n$ denotes the Euclidean space of dimension $n$. For any matrix $K$, $K^T$ and $K^{-1}$ denotes its transpose and inverse respectively. $\lambda_m(\cdot)$ and $\lambda_M(\cdot)$ denotes the minimum and maximum Eigen values of a given matrix respectively. We denote $\|\cdot\|$ as the standard Euclidean 2-norm of a given vector and $\log(\cdot)$ represents the logarithmic function.
	\par Consider a general nonlinear dynamical system:
	\begin{equation}
		\dot x(t) = f(x(t)), \ x(0)=x_0 \label{eqn1}
	\end{equation}
	where, $f:\mathcal{D}\rightarrow\mathbb{R}^n$ be a $\mathcal{C}^0$ (continuous) function with $\mathcal{D}\subset\mathbb{R}^n$ being an open neighbourhood of the origin and $f(0)=0$. This implies that $f^{-1}(0)\triangleq\{x\in\mathcal{D}:f(x)=0\}$ is non-empty. A $\mathcal{C}^1$ (continuously differentiable) function $x:\mathcal{I}\rightarrow\mathcal{D}$ is a solution of \eqref{eqn1} on the interval $\mathcal{I}\subset\mathbb{R}$, if $x$ satisfies \eqref{eqn1} $\forall t \in \mathcal{I}$. The function $f$ being continuous implies that $\forall$ $x_0\in\mathcal{D}, \ \exists \ t_0< 0<t_1$ and a solution $x(\cdot)$ of \eqref{eqn1} defined on $t\in(t_0,t_f)$ such that $x(0)=x_0$. Suppose for all initial conditions, \eqref{eqn1} possess unique solution except at the origin \citep{bohn2016almost}. This implies that for every $x\in\mathcal{D}\setminus \{0\}, \exists \ t_2 > 0$ and a unique solution $x(\cdot)$ of \eqref{eqn1} defined on $[0,t_2)$ satisfying $x(0)=x_0$.  
	\begin{proposition}\citep{ref23}
		Let $\mathcal{O}\subseteq\mathcal{D}$ be an open neighbourhood of origin and let $\mathcal{T}(x):\mathcal{O}\backslash\{0\} \rightarrow \mathbb{R}^+$ be the finite time. Suppose  the following statements hold.
		\begin{enumerate}
			\item[a)] For every $x\in\mathcal{O}\backslash\{0\}$, a unique solution $\kappa_t(x)$ is defined for $t\in[0,\mathcal{T}(x)]$, $\kappa_t(x)\in\mathcal{O}\backslash\{0\}$ and $\displaystyle{\lim_{t \to \mathcal{T}(x)} \kappa_t(x) = 0}$ implies finite-time convergence.
			\item[b)] Suppose for every open set $\mathcal{O}_\epsilon: 0\in\mathcal{O}_\epsilon\subseteq\mathcal{O}, \exists$ an open set $\mathcal{O}_\delta:0\in\mathcal{O}_\delta\subseteq\mathcal{O}$ such that for every $x\in\mathcal{O}_\delta\backslash\{0\}, \kappa_t(x) \in \mathcal{O}_\epsilon \ \forall \ t\in[0,\mathcal{T}(x)]$. This yields Lyapunov stability.
		\end{enumerate}
		Then, the origin is said to be a finite-time stable equilibrium of \eqref{eqn1}. It will be globally finite-time stable if $\mathcal{D}=\mathcal{O}=\mathbb{R}^n$.
	\end{proposition}
	\begin{lemma}[\citep{ref23}]\label{lem1}
		Let $c>0$ be a positive real number, $\alpha\in(0,1)$ and $V:\mathcal{D}\rightarrow\mathbb{R}$ be a $\mathcal{C}^1$ function. Suppose $\exists$ a neighborhood $\mathcal{N}\subset\mathcal{D}$ such that $V$ is positive definite on $\mathcal{N}$ and $\dot{V}+cV^\alpha$ is negative semi-definite on $\mathcal{N}$, where $\dot{V}(x)= \frac{\partial V}{\partial x}f(x)$. Then, the origin is said to be finite-time stable equilibrium of \eqref{eqn1}. Furthermore, if $\mathcal{T}(x)$ is the finite-time, then $\mathcal{T}(x)\leq\frac{1}{c(1-\alpha)}V(0)^{1-\alpha} \forall x\in\mathcal{N}$. 
	\end{lemma}
	%	Suppose there exists a continuously differentiable function $V\colon\ \mathcal{D}\to\mathbb{R} $,  $ k > 0$ (a real number) and $ \alpha \in (0,1)$ and a
	%	neighborhood $U\subseteq D$ of the origin such that V is positive definite on $U$ and  $\dot{V}+kV^{\alpha}$ is negative semi-definite on $U$, where $\dot {V}$ = $ \frac{\partial v}{\partial x}(x)f(x) $ then, 
	%	\begin{itemize}
		%		\item The origin is a finite-time-stable equilibrium of eq. \eqref{eqn1}.
		%		\item If $T$  is the settling time ,then $T(x) \leq \frac{1}{k(1-\alpha)}V(x)^{(1-\alpha)}$ for all $x$ in some open neighbourhood of the origin.
		%	\end{itemize}
	
	\begin{lemma}[\citep{khalil2002nonlinear} Rayleigh-Ritz inequality]\label{ritz theorem} Consider a symmetric positive definite matrix, $P \in \mathbb{R}^{n \times n}$ and let $\lambda\in\mathbb{R}^{n}$ be the eigenvalues of $P$. Then the following inequality,
		\begin{equation} \label{ritz positive sign}
			\lambda_{m}(P)\left\|{x}\right\|^2 \leq x^{T}Px \leq \lambda_{M}(P)\left\|{x}\right\|^2  \quad \forall \ x\in \mathbb{R}^{n},
		\end{equation}
		holds, where, $\lambda_{m}(\cdot)$ and $\lambda_{M}(\cdot)$ denotes the minimum and maximum eigenvalues respectively.
	\end{lemma}
	\begin{lemma}\label{lemma inequality}
		Suppose $a_1$, $a_2$ are positive real numbers and let $c \in(0,1)$, then the following inequality holds: 
		\begin{equation*}
			(a_1+a_2)^c \leq(a_1^c + a_2^c).
		\end{equation*}
	\end{lemma}	
	\begin{proof}
		To aid the proof, we make use of the subadditivity property of a concave function. \newline \textit{Subadditivity of a function}:  A function $f:\mathbb{R}_{\geq 0} \supseteq \mathcal{U}\rightarrow\mathbb{R}$ is said to be subadditive on $\mathcal{U}$ if $\forall a_1,a_2 \in \mathcal{U}$ such that $a_1+a_2\in\mathcal{U}$, the following inequality holds:
		$f(a_1+a_2)\leq f(a_1)+f(a_2).$
		\newline \textit{Concavity}: Let $c\in[0,1]$ and $a_1,a_2 \in \mathcal{U}$, a concave function $f:\mathcal{U}\rightarrow\mathbb{R}$ with $f(0)\geq0$ satisfying,
		\begin{equation}\label{subadd}
			f(ca_1+(1-c)a_2) \geq cf(a_1)+(1-c) f(a_2)
		\end{equation}
		is subadditive. Choose $a_2=0$ in \eqref{subadd}, this yields $f(ca_1)\geq cf(a_1)$. Suppose $c=\frac{a_1}{a_1+a_2}$, then we have,
		\begin{equation*}\label{Aeq1}
			f(a_1)=f(c(a_1+a_2))\geq cf(a_1+a_2)
		\end{equation*}
		\begin{equation*}\label{Aeq2}
			f(a_2)=f((1-c)(a_1+a_2))\geq (1-c)f(a_1+a_2).
		\end{equation*}
		Adding the above two equations yields,
		\begin{equation*}
			f(a_1+a_2) \leq f(a_1)+f(a_2)\implies \mbox{subadditivity}.
		\end{equation*}
		By considering $f(a_1)=a_1^c$, $f(a_2)=a_2^c$ and $f(a_1+a_2)=(a_1+a_2)^c$, yields the proof.
	\end{proof}

	\section{Mathematical model of FWMR} \label{sec:3}
	Consider  a Four Wheeled Mecanum Robot (FWMR) model shown in Figure  \ref{fig:FMWMR}. 
	\begin{figure}[htpb!]
		\centering
		\includegraphics[scale=0.32]{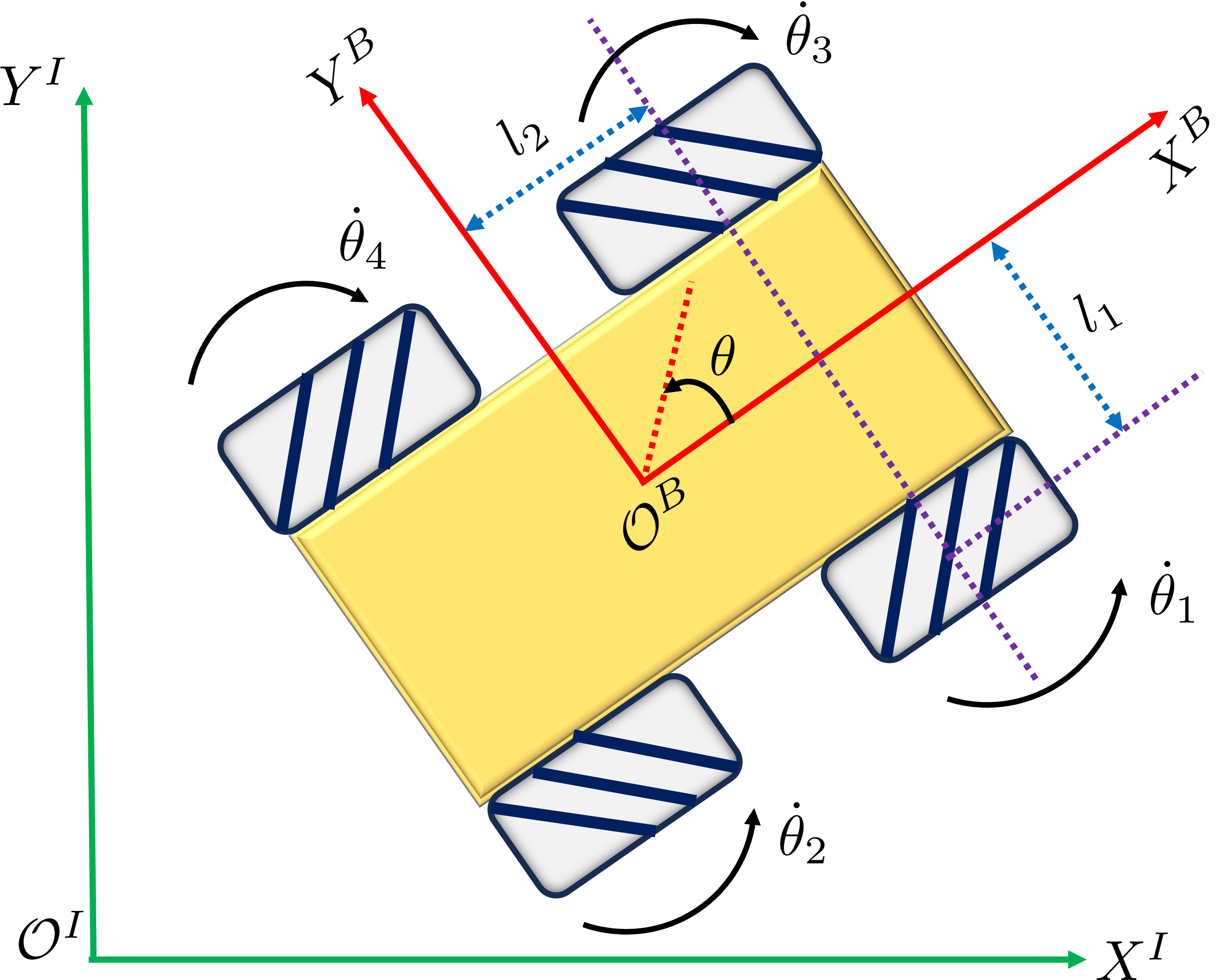}
		\caption{Schematic representation of a Four  Wheeled Mecanum Robot.}
		\label{fig:FMWMR}
	\end{figure}
	It forms a unique type of omnidirectional mobile robot, characterized by a rectangular configuration with four specialized mecanum wheels placed symmetrically at the corners. Each mecanum wheels are equipped with a series of passive rollers along its circumference. These rollers are typically oriented with their axis of rotation at an angle of $45^\circ$ with the wheel roll axis. This distinctive design affords FWMR an additional degree of freedom compared to the conventional wheeled mobile robots. Consequently, FWMR possess the capability to maneuver in any direction without necessitating reorientation.
	It contains four mecanum wheels and $\dot \theta_i$ is the angular velocity of the robot $\mathcal{O}^B$. $l_1$ and $l_2$ are the lengths from the center of mass and center of wheels of FWMRs.
	In order to describe the pose of the vehicle, we consider two coordinate frames which are mutually orthonormal to each other. Suppose $(X^B,Y^B)$ denotes coordinates of position expressed in vehicle fixed frame $\mathcal{O}^B$ located at the center of mass of the vehicle. Let  $(X^I,Y^I)$ be the corresponding coordinates expressed in earth fixed frame $\mathcal{O}^I$ as shown in Figure \ref{fig:FMWMR}. Suppose $q =(x,y,\theta)\in\mathcal{D}\subset\mathbb{R}^2\times[0,2\pi)$ denotes the pose of the vehicle expressed in $\mathcal{O}^I$ frame and let $\dot q = (\dot x,\dot y,\dot \theta)\in\mathbb{R}^3$ denotes the corresponding inertial fixed velocity of the robot. The transformation which relates the velocities in $\mathcal{O}^I$ and $\mathcal{O}^B$ frames is given by,  
	\begin{equation} \label{relation in speeds}
		\dot q^{B} = Q \dot q, \quad Q = \begin{bmatrix}
			\cos\theta& \sin\theta & 0\\
			-\sin\theta & \cos\theta & 0\\
			0   &      0  &     1
		\end{bmatrix},
	\end{equation}
	where, $Q\in \mathbb{R}^{3\times3}$ represents the transformation matrix and  $\dot q^{B}$ represents the robot velocity in $\mathcal{O}^B$ frame.
	Let $ \rho  = (\theta_1, \theta_2, \theta_3, \theta_4)$ represents the angular position of the wheels. 
	Wheel coordinate velocities can be related to body coordinate velocities using simple geometrical relationship from Figure \ref{fig:FMWMR} as,
	$
	\dot q^{B} = J\dot{\rho}, \quad J = \frac{r}{4}\begin{bmatrix}
		-1 & 1 & -1 & 1\\
		1  & 1 & 1 & 1 \\
		\frac{1}{l_1+l_2} & -\frac{1}{l_1+l_2}  & -\frac{1}{l_1+l_2}  & \frac{1}{l_1+l_2}
	\end{bmatrix},
	$
	where, $J$ is the transformation matrix. $l_1$  and $l_2$ are the distance of center of mass of the robot from the respective wheel centers as shown in Figure \ref{fig:FMWMR}.
	Wheel coordinate velocities are given by,
	\begin{equation}
		\dot \rho =  J^\dag\dot q^{B} \label{rho dot}, \quad  J^\dag =  \begin{bmatrix}
			-1  & 1  & l_1+l_2\\
			1  & 1  & -l_1-l_2\\
			-1  &  1  & -l_1-l_2\\
			1   &  1   & l_1+l_2
		\end{bmatrix},
	\end{equation}
	where,  $J^\dag$ denotes the right pseudo inverse of $J$. 
	Dynamical model of FWMR obtained using Lagrange D'Alembert's formulation \citep{yuan2020trajectory,8267134,zhao2022fixed} is given by,  
	\begin{equation} \label{equ of motion in rho}
		M\ddot{\rho}+D\dot{\rho} = \tau + \tau_d,
	\end{equation}
	\begin{equation*}
		\mbox{where,} \ M = \begin{bmatrix}
			\Gamma & -m_1  & m_1 & m_2-m_1\\
			-m_1  & \Gamma  &  m_2-m_1  & m_1\\
			m_1 & m_2-m_1 & \Gamma & -m_1\\
			m_2-m_1  & m_1  & -m_1  & \Gamma\\
		\end{bmatrix},
	\end{equation*}
	$m_1 =\frac{I_zr^2}{16(l_1+l_2)^2}, \ m_2 =\frac{mr^2}{8}, \ D = \mbox{diag}(D_1, D_2, D_3, D_4)$ and $\Gamma=m_1+m_2+I_m$.
	Control input, $\tau \triangleq (\tau_1, \tau_2, \tau_3, \tau_4)\in\mathbb{R}^4$ represents the torques acting on four wheels of the robot; $I_m$ and $I_z$ denotes the moment of inertia of wheels and the robot respectively; $m$ denotes mass of the robot; $r$ is the wheel radius; $D_i$'s are viscous friction coefficients of the four wheels; $\tau_d\in\mathbb{R}^4$ forms the external disturbances including static friction. Using kinematics \eqref{relation in speeds} and  \eqref{rho dot}, we get the acceleration relationship as,
	\begin{equation} \label{ddot rho}
		\ddot \rho = J^\dag Q\ddot q=J^\dag \ddot q^B,
	\end{equation}
	using \eqref{relation in speeds}  and \eqref{ddot rho} in \eqref{equ of motion in rho} yields,
	\begin{equation} \label{equ of motion in q}
		MJ^\dag \ddot q^B + DJ^\dag \dot q^B = \tau + \tau _d,
	\end{equation}
	To make ease of control synthesis, we apply a transformation to convert the above four-dimensional torque expressed in wheel coordinates into  body fixed torque in $\mathbb{R}^3.$
	Multiply by $JM^{-1}$ in \eqref{equ of motion in q} and rearranging yields,
	\begin{equation} \label{dyanmics after multiplying}
		\ddot q^B + JM^{-1}DJ^\dag Q\dot q^B = JM^{-1}\tau +JM^{-1} \tau_d,
	\end{equation}
	which can be expressed as,
	\begin{equation} \label{final dynamics}
		\ddot q^B =- JM^{-1}DJ^\dag Q\dot q^B +JM^{-1}\tau +JM^{-1} \tau_d.
	\end{equation}
	For notational simplicity, let $\eta=q$ and $\nu=\dot q^B$. Then, the second-order unperturbed ($\tau_d=0$) dynamical system can be expressed by two first-order systems as,
	\begin{equation}\label{Kinamtics and dynamics with error}
		\begin{aligned}
			\dot \eta =  \bar{Q}\nu, \quad   
			\dot \nu = F(\eta)\nu+\bar{\tau} ,
		\end{aligned}
	\end{equation}
	where, $\bar{Q}=Q^T$, $ F(\eta)=-JM^{-1}DJ^\dag Q$ and $\bar{\tau}=JM^{-1}\tau$. One can observe that the dimension of control input is now reduced from four to three which makes the subsequent controller synthesis easier. \newline These two equations governing the kinematics and dynamics of FWMR represents the complete system dynamics which serves as the starting point for the synthesis of finite-time tracking control law which is discussed in the next section.
	\begin{assumption}\label{assum1}
		The desired center of mass position and orientation of FWMR, $\eta_d(t)$ is $\mathcal{C}^2$ with its first and second derivatives are bounded.
	\end{assumption}
	\begin{problem}
		Consider the dynamics \eqref{Kinamtics and dynamics with error}, with the feedback law $\bar{\tau}$. Given a desired trajectory $\eta_d(t)$ satisfying \textit{Assumption \ref{assum1}}, our aim is to design the feedback law $\bar{\tau}$, such that the vehicle pose error, $\tilde{\eta}(t)\triangleq \eta(t)- \eta_d(t)$  converges to the origin in finite-time.
	\end{problem}

	\section{Theoretical developments} \label{sec:4}
	With the complete dynamical  model of FWMR described in section \ref{sec:3}, designing a finite-time tracking control law using backstepping technique is outlined in this section.  A schematic representation of the proposed control architecture is shown in Figure \ref{fig:block}. 
	\begin{figure*}[htpb!]
		\centering
		\includegraphics[scale=0.35]{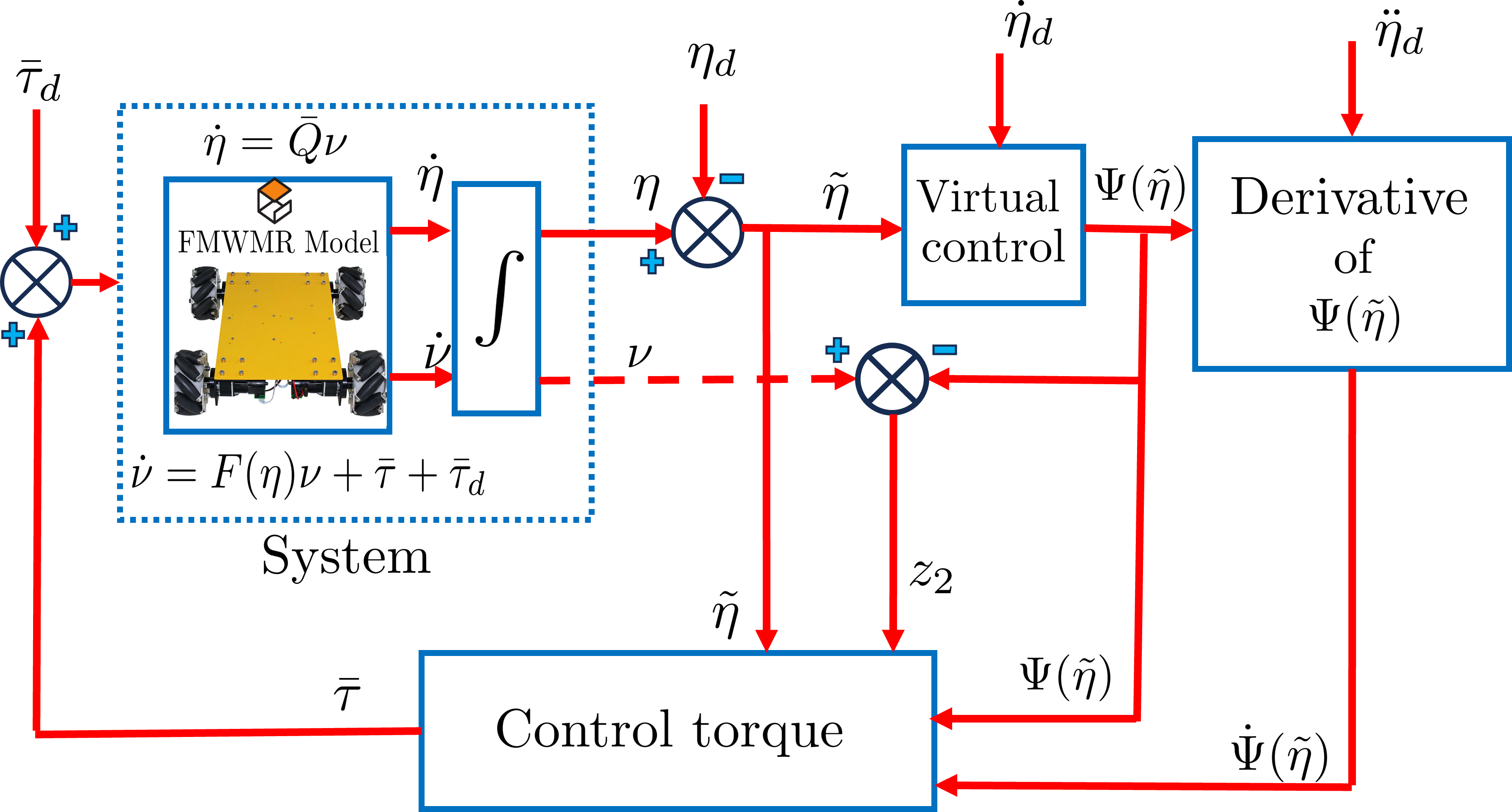}
		\caption{Schematic representation of the proposed control architecture.}
		\label{fig:block}
	\end{figure*}
	Let  $\eta_d\in\mathcal{D}$ denotes desired pose (position and orientation) of the vehicle in $\mathcal{O}^I$ frame. Suppose $\nu_d\in \mathbb{R}^3$ and $\dot{\eta}_d\in \mathbb{R}^3$ be the desired velocities of the robot  expressed in $\mathcal{O}^B$ and $\mathcal{O}^I$ frames respectively. We define the tracking error for pose of the vehicle as, $\tilde{\eta}\triangleq\eta-\eta_d\in\mathcal{D}$ satisfying assumption \ref{assum1} . Consider a Lyapunov function $V_1:\mathcal{D}\rightarrow\mathbb{R}$ for the configuration error subsystem,
	\begin{equation} \label{first Lyapunov funtion}
		\begin{aligned}
			V_1(\tilde{\eta}) = \frac{1}{2}\tilde{\eta}^{T}\tilde{\eta}.
		\end{aligned}
	\end{equation}
	Taking the time derivative of $V_1(\tilde{\eta})$ yields,
	\begin{equation}\label{virtual}
		\dot{V}_1(\tilde{\eta})=\tilde{\eta}^T\dot{\tilde{\eta}}={\tilde{\eta}}^{T}( {\dot{\eta} - \dot{\eta}_d})={\tilde{\eta}}^{T}( {Q{\nu} - \dot{\eta}_d}),
	\end{equation}
	Next we choose the virtual control, $\Psi(\tilde{\eta})\in\mathcal{C}^1$ for position error subsystem as,
	\begin{equation}\label{virtualc}
		\nu^\ast\triangleq\Psi(\tilde{\eta})= \bar{Q}^{T} \left(\dot{\eta}_d-\frac{K_\eta\tilde{\eta}}{\left\|{\tilde{\eta}}\right\|^{(1-\alpha)}}\right),
	\end{equation}
	where, $K_\eta\in\mathbb{R}^{3\times3}$ is a symmetric positive definite matrix and the parameter $\alpha$ is further chosen (refer, \textit{Lemma} \ref{lemma1}) to ensure the boundedness of feedback torque input designed through backstepping procedure. We denote the virtual velocity error for velocity subsystem as, $z_2\triangleq\nu-\Psi(\tilde{\eta})$. Thus in terms of transformed error coordinates, the system \eqref{Kinamtics and dynamics with error} can be re-expressed as,
	\begin{equation}\label{Kinamtics and dynamics with error1}
		\begin{aligned}
			\dot {\tilde{\eta}}& =  Q(z_2+\Psi(\tilde{\eta}))-\dot\eta_d,\\  
			\dot z_2&= F(\tilde{\eta})(z_2+\Psi(\tilde{\eta}))-\dot{\Psi}(\tilde{\eta})+\bar{\tau} .
		\end{aligned}
	\end{equation}
	
	%Putting $\nu\triangleq\Psi(\tilde{\eta})$ in \eqref{virtual} yields,
	%\begin{equation}
	%	\begin{aligned}
		%		\dot V_1(\tilde{\eta}) = -{\left(\frac{{\tilde{\eta}}^{T}K_\eta\tilde{\eta}}{\left\|\tilde{\eta}\right\|^{(1-\alpha)}}\right)}.
		%	\end{aligned}
	%\end{equation}
	Suppose $\lambda_m(K_\eta)$ denotes the minimum eigenvalue of $K_\eta$, then putting \eqref{virtualc} in \eqref{virtual} yields,
	\begin{equation}\label{ineq1}
		\begin{aligned}
			\dot V_1(\tilde{\eta}) \leq& -{\lambda_m(K_\eta)}\frac{\tilde{\eta}^{T}\tilde{\eta}}{\left\|\tilde\eta\right\|^{(1-\alpha)}}
			\leq -{\lambda_m(K_\eta)}\frac{\left\|\tilde\eta\right\|^{2}}{\left\|\tilde\eta\right\|^{(1-\alpha)}},\\
			\leq&-{\lambda_m(K_\eta)}{\left\|\tilde\eta\right\|^{(1+\alpha)}}
			\leq-{\lambda_m(K_\eta)}\left({\left\|\tilde\eta\right\|^{2}}\right)^{\frac{(1+\alpha)}{2}},\\
			\leq& -{\lambda_m(K_\eta)}(2V_1)^{\frac{(1+\alpha)}{2}}\leq-c_1V_1^\beta,
		\end{aligned}
	\end{equation}
	where, $c_1={\lambda_m(K_\eta)}2^\beta$ and $\beta=\frac{(1+\alpha)}{2}$. Therefore, by making use of $\textit{Lemma \ref{lem1}}$, configuration error subsystem, $\tilde{\eta}$ is finite-time stable about the origin with the convergence time  calculated as,
	\begin{equation}
		\mathcal{T}_\eta\leq\frac{V_1^{1-\beta}(\tilde{\eta}(0))}{c_1{(1-\beta)}},
	\end{equation}
	where,  $\tilde{\eta}(0)$ is the initial condition of $\tilde{\eta}$ at $t=0$, $\mathcal{T}_\eta$ denotes the finite-time for configuration error subsystem and $V_1(\tilde{\eta}(0))$ represents the Lyapunov function evaluated at  $t=0$.
	\begin{lemma} \label{lemma1}
		Range of $\alpha$ is $\frac{1}{2}<\alpha<1$ for the system \eqref{Kinamtics and dynamics with error1}.
	\end{lemma}
	\begin{proof}
		To see this, consider the derivative of virtual error function, 
		\begin{align}		 \label{tdpsi}
			\dot{\Psi}(\tilde{\eta}) &= \dot{\bar{Q}}^T\left(\dot{\eta}_d - \frac{K_\eta \tilde{\eta}}{\left\|{\tilde{\eta}}\right\|^{1-\alpha} }\right) + \bar{Q}^T\left(\ddot{\eta}_d - \frac{K_\eta \dot{\tilde{\eta}}}{\left\|{\tilde{\eta}}\right\|^{1-\alpha}}\right) \\ \nonumber
			&\quad - \bar{Q}^T\left(\frac{K_\eta \tilde{\eta}\tilde{\eta}^T\dot{\tilde{\eta}}(\alpha-1)}{\left\|{\tilde{\eta}}\right\|^{3-\alpha}}\right).
		\end{align}
		Under the limiting case, $\tilde{\eta} \to 0$ and $z_2 \to 0$ on the function $\dot{\Psi}(\tilde{\eta})$,  L-Hospital's rule is applied on the second term of \eqref{tdpsi}. From this, and \eqref{virtualc} yields the range of $\alpha>0$, while the limiting case applied on the fourth and fifth terms of \eqref{tdpsi} gives the range of $\alpha>\frac{1}{2}$ for ensuring the boundedness of  $\dot{\Psi}(\tilde{\eta})$. Combining these two yields, $\alpha>\frac{1}{2}$ and hence the proof.
	\end{proof}
	\subsection{Stability Analysis}
	Next, we explain the development of proposed feedback torque designed using backstepping principles and the closed loop stability analysis with the controller yielding finite-time global convergence of states to the equilibrium.
	\begin{theorem}\label{theo1}
		Consider the reference trajectory $\eta_d(t)\in \mathcal{D}$ satisfying \textit{Assumption \ref{assum1}}, and the transformed closed loop dynamical model of FWMR \eqref{Kinamtics and dynamics with error1}.  With the feedback torque defined as,
		\begin{equation}\label{actual controller}
			\bar{\tau} = -\bar{Q}\tilde{\eta}-F(\tilde{\eta})(z_2+\Psi(\tilde{\eta}))+\dot\Psi(\tilde{\eta})-\frac{K_zz_2}{\left\|{z_2}\right\|^{(1-\alpha)}},
		\end{equation}
		where, $K_z\in\mathbb{R}^{3\times3}$ is a positive definite matrix and $\dot\Psi(\tilde{\eta})$ is defined in \eqref{tdpsi}, the resulting closed loop system \eqref{Kinamtics and dynamics with error1} is globally finite-time stable around the equilibrium $(\tilde{\eta},z_2)=(0,0)$. Consequently, the states $(\eta,\nu)$ will converge to the desired trajectory $(\eta_d,\nu_d)$.
	\end{theorem}
	\begin{proof}
		To obtain backstepping controller, consider the Lyapunov function for velocity error subsystem $z_2$ as $V_2:\mathbb{R}^3\rightarrow\mathbb{R}$ given by,
		\begin{equation*}
			V_2(z_2)=\frac{1}{2}z_2^Tz_2.
		\end{equation*}
		Taking the time derivative yields,
		\begin{equation} \label{derivative of second lyapunov}
			\dot V_2= {z}^{T}_2\dot {z_2}={z}^{T}_2{({\dot{\nu} -\dot\Psi({\tilde{\eta}}))}}.
		\end{equation}
	Substituting the dynamics \eqref{Kinamtics and dynamics with error1} gives,
		\begin{equation}\label{final dot V_2}
			\begin{aligned}
				\dot V_2 =& {z}^{T}_2{\left(F(\tilde{\eta})(z_2+\Psi({\tilde{\eta}}))+\bar{\tau}-\dot\Psi({\tilde\eta})\right)}.
			\end{aligned}
		\end{equation}
		Next we choose a Lyapunov function for the total system, $V:\mathcal{D}\times\mathbb{R}^3\rightarrow\mathbb{R}$ as $V(\tilde{\eta},z_2)=V_1(\tilde{\eta})+V_2(z_2)$. Taking the time derivative yields,
		%\begin{equation*}
		%	\dot V(\tilde{\eta},z_2)=\dot V_1(\tilde{\eta})+\dot V_2(z_2),
		%\end{equation*}
		\begin{equation*}
			\begin{aligned}
				\dot V(\tilde{\eta},z_2) &=\dot V_1(\tilde{\eta})+\dot V_2(z_2), \\
				& =\tilde{\eta}^T(\bar{Q}(z_2+\Psi(\tilde{\eta}))-\dot\eta_d)  \\ & \ \  \ +{z}^{T}_2{\left(F(\tilde{\eta})(z_2+\Psi({\tilde{\eta}}))+\bar{\tau}-\dot\Psi({\tilde\eta})\right)}.
			\end{aligned}
		\end{equation*}
		Choosing the feedback torque,
		\begin{equation}\label{actual control}
			\bar{\tau}\triangleq-\bar{Q}\tilde{\eta}-F(\tilde{\eta})(z_2+\Psi({\eta}))+\dot\Psi(\tilde\eta)-\frac{K_zz_2}{\left\|{z_2}\right\|^{(1-\alpha)}},
		\end{equation}
		and substituting it in \eqref{final dot V_2} yields,
		\begin{equation}
			\dot V_2 = -\frac{{z_2}^{T}K_zz_2}{\left\|{z_2}\right\|^{(1-\alpha)}}.
		\end{equation}
		From \textit{Lemma} \ref{lemma1}, feedback torque \eqref{actual controller} is bounded for the prescribed range of $\alpha$. Suppose, $\lambda_{m}(K_z)$ denotes the minimum eigenvalue of $K_z$. By making use of \textit{Lemma \ref{ritz theorem}}, 
		\begin{equation}
			\begin{aligned}
				\dot V_2 \leq& -{\lambda_{m}(K_z)}\frac{z_2^{T}z_2}{\left\|z_2\right\|^{(1-\alpha)}}
				\leq-{\lambda_{m}(K_z)}\frac{\left\|z_2\right\|^{2}}{\left\|z_2\right\|^{(1-\alpha)}}, \\
				\leq&-{\lambda_{m}(K_z)}{\left\|z_2\right\|^{(1+\alpha)}}
				\leq-{\lambda_{m}(K_z)}\left({z_2}^{T}{z_2}\right)^{\frac{(1+\alpha)}{2}}, \\
				\leq& -{\lambda_{m}(K_z)}\left({2V_2}\right)^{\frac{(1+\alpha)}{2}}
				\leq -{\lambda_{m}(K_z)}2^{\frac{1+\alpha}{2}} \left({V_2}\right)^{\frac{(1+\alpha)}{2}},\\
				% =& -2b_1{\lambda_1}^{\frac{1-\alpha}{2}}(2V_1)^{\frac{1+\alpha}{2}}\\
				% =&-2^{\frac{(1+\alpha)}{2}}b_1{\lambda_1}^{\frac{1-\alpha}{2}}(V_1)^{\frac{1+\alpha}{2}}
			\end{aligned}
		\end{equation}
		where, in the last inequality, taking $c_2={\lambda_m(K_z)}2^\beta$ and $\beta=\frac{(1+\alpha)}{2}$ yields,
		\begin{equation} \label{final Dot V_2}
			\dot V_2 \leq -c_2V_2^{\beta}.
		\end{equation}
		Using inequalities \eqref{ineq1} and \eqref{final Dot V_2},
		\begin{equation}
			\dot V(\tilde{\eta},z_2)\leq-c_1V_1^\beta-c_2V_2^\beta\leq-c(V_1^\beta+V_2^\beta),
		\end{equation}
		where, it is straightforward to see that the above inequality will be always satisfied with the choice of $c=\mbox{min}(c_1,c_2)$ for any initial states $(\tilde{\eta}(0),\nu(0))$. Next, we make use of \textit{Lemma \ref{lemma inequality}} which yields,
		$
		\dot V(\tilde{\eta},z_2)\leq-c(V_1^\beta+V_2^\beta)\leq-c(V_1+V_2)^\beta\leq-cV^\beta.
		$
		Therefore, by using \textit{Lemma \ref{lem1}}, the overall system represented by \eqref{Kinamtics and dynamics with error1} with the closed loop feedback torque \eqref{actual control} is globally finite-time stable with the guaranteed convergence time,
		\begin{equation}\label{time}
			\mathcal{T}\leq \frac{V_0^{(1-\beta)}}{c{(1-\beta)}},
		\end{equation}
		where, $V_0\triangleq V(\tilde{\eta}(0),z_2(0))$ is the Lyapunov function evaluated at  time $t=0$. This yields the proof.
	\end{proof}
	
	\begin{corollary} 
		Control law \eqref{actual controller} is continuous and gives finite-time global convergence of the states to the equilibrium for $\alpha\in(\frac{1}{2},1]$.
	\end{corollary}
	
	\begin{proof}
		Continuity of the control law \eqref{actual controller} for  $\alpha\in(\frac{1}{2},1)$ can be seen directly from \textit{Lemma} \ref{lemma1}. In order to show that the continuity also holds for $\alpha=1$, evaluate the partial derivative of control input $\bar{\tau}$ \eqref{actual controller} with respect to $\alpha$ and then equate it to the limiting value of $\alpha=1$. Computing the partial derivative of \eqref{actual controller} with respect to $\alpha$ yields,
		
		%	\begin{align*}
			%			\frac{\partial\bar{\partial\tau}}{\alpha}&=\frac{\partial}{\partial \alpha}\left(-\bar{Q}\tilde{\eta}-F(\tilde{\eta})\left(z_2+\bar{Q}^{T} \left(\dot{\eta_d}-\frac{K_\eta\tilde{\eta}}{\left\|{\tilde{\eta}}\right\|^{(1-\alpha)}}\right)\right) \\ \nonumber
			%		&\quad \qquad+ \dot{\bar{Q}}^T\left(\dot{\eta}_d - \frac{K_\eta \tilde{\eta}}{\left\|{\tilde{\eta}}\right\|^{1-\alpha} }\right) + \bar{Q}^T\left(\ddot{\eta}_d - \frac{K_\eta \dot{\tilde{\eta}}}{\left\|{\tilde{\eta}}\right\|^{1-\alpha}}\right) \\ \nonumber
			%		& - \bar{Q}^T\left(\frac{K_\eta \tilde{\eta}\tilde{\eta}^T\dot{\tilde{\eta}}(\alpha-1)}{\left\|{\tilde{\eta}}\right\|^{3-\alpha}}\right)-\frac{K_zz_2}{\left\|{z_2}\right\|^{(1-\alpha)}}\right)
			%	\end{align*}
		\begin{align*}
			\frac{\partial \bar{\tau}}{\partial \alpha} &= \frac{\partial}{\partial \alpha} \left( -\bar{Q} \tilde{\eta} - F(\tilde{\eta}) \left( z_2 + \bar{Q}^T \left( \dot{\eta}_d - \frac{K_\eta \tilde{\eta}}{\|\tilde{\eta}\|^{1-\alpha}} \right) \right) \right. \\
			&\quad \left. + \dot{\bar{Q}}^T \left( \dot{\eta}_d - \frac{K_\eta \tilde{\eta}}{\|\tilde{\eta}\|^{1-\alpha}} \right) + \bar{Q}^T \left( \ddot{\eta}_d - \frac{K_\eta \dot{\tilde{\eta}}}{\|\tilde{\eta}\|^{1-\alpha}} \right) \right. \\
			&\quad \left. - \bar{Q}^T \left( \frac{K_\eta \tilde{\eta} \tilde{\eta}^T \dot{\tilde{\eta}} (\alpha - 1)}{\|\tilde{\eta}\|^{3 - \alpha}} \right) - \frac{K_z z_2}{\|z_2\|^{1-\alpha}} \right).
		\end{align*}
		Taking the terms containing $\alpha$ yields,
		\begin{align*}
			\frac{\partial \bar{\tau}}{\partial \alpha} &= \frac{\partial}{\partial \alpha} \left( \left(F(\tilde{\eta})\bar{Q}^T\tilde{\eta}-\dot{\bar{Q}}^T\tilde{\eta}-\bar{Q}^T\dot{\tilde{\eta}}\right)K_\eta \|\tilde{\eta}\|^{\alpha-1}\right. \\ 
			&\left. -\bar{Q}^TK_\eta\tilde{\eta}\tilde{\eta}^T\dot{\tilde{\eta}}(\alpha-1)\|\tilde{\eta}\|^{\alpha-3}-K_z z_2\|z_2\|^{\alpha-1}\right). 
		\end{align*}
		Upon simplification yields,
		\begin{align*}
			\frac{\partial \bar{\tau}}{\partial \alpha} &= K_\eta\|\tilde{\eta}\|^{\alpha-1}\log(\|\tilde{\eta}\|)\left( F(\tilde{\eta})\bar{Q}^T\tilde{\eta}-\dot{\bar{Q}}^T\tilde{\eta}-\bar{Q}^T\dot{\tilde{\eta}}\right)\\  &\quad -\bar{Q}^TK_\eta\tilde{\eta}\tilde{\eta}^T\dot{\tilde{\eta}}\left((\alpha-1)\|\tilde{\eta}\|^{\alpha-3}\log(\|\tilde{\eta}\|)+\|\tilde{\eta}\|^{\alpha-3}\right) \\ &\quad -K_z z_2\|z_2\|^{\alpha-1} \log(\|z_2\|)
		\end{align*}
		Subsequently, evaluating the above derivative for the limit $\alpha \leadsto 1$ yields,
		\begin{align*}
			\lim_{\alpha\to1} \frac{\partial \bar{\tau}}{\partial \alpha} &=K_\eta\log(\|\tilde{\eta}\|)\left(F(\tilde{\eta})\bar{Q}^T\tilde{\eta}-\dot{\bar{Q}}^T\tilde{\eta}-\bar{Q}^T\dot{\tilde{\eta}}\right) \\ & -\bar{Q}^TK_\eta\dot{\tilde{\eta}}-K_z z_2\log(\|z_2\|),
		\end{align*}
		which shows that the control law is continuous in the range $\alpha\in(\frac{1}{2},1)$ and it also holds when $\alpha=1$.  This signifies that the vector field on $\mathcal{D}\times\mathbb{R}^3$ given by the dynamics \eqref{Kinamtics and dynamics with error1} under the influence of control torque \eqref{actual controller} is continuous with respect to the parameter $\alpha$ in the prescribed range $(\frac{1}{2},1]$. Therefore, the flow of closed loop feedback system is continuous with respect to $\alpha\in(\frac{1}{2},1]$. Hence, integrating this with the global finite-time convergence result of \textit{Theorem} \ref{theo1} yields the proof.
	\end{proof}
	\par Note that, $\alpha=1$ represents the asymptotic controller. We analyzed the continuity of $\bar{\tau}$ at $\alpha=1$ to express it within a general framework such that it can be derived as a subcase of the proposed finite-time controller.
	\subsection{Robustness Analysis}
	For the unperturbed system \eqref{Kinamtics and dynamics with error1}, we have proved the finite-time global convergence of system states to the desired equilibrium. For demonstrating the robustness of the proposed controller in the presence of disturbances, we recast the dynamical equation which follows from \eqref{final dynamics} as,
	\begin{equation}\label{Kinamtics and dynamics with error2}
		\begin{aligned}
			\dot {\tilde{\eta}}& =  Q(z_2+\Psi(\tilde{\eta}))-\dot\eta_d,\\  
			\dot z_2&= F(\tilde{\eta})(z_2+\Psi(\tilde{\eta}))-\dot{\Psi}(\tilde{\eta})+\bar{\tau}+\bar{\tau}_d,
		\end{aligned}
	\end{equation} 
	where, $\bar{\tau}_d$ denotes the constant bounded disturbance torque acting on the system. We explicitly provide an expression for the maximum norm bound of the disturbance torque expressed using a desired side of neighbourhood about the equilibrium point that can be tolerated for the system states to yield the finite-time global convergence result elucidated in \textit{Theorem} \ref{theo1}.
	\begin{corollary}
		Consider the closed loop perturbed system \eqref{Kinamtics and dynamics with error2} with the control torque \eqref{actual controller}. Suppose, $\mathcal{N}\subset\mathcal{D}\times\mathbb{R}^3$ be a closed neighbourhood of the origin defined by,
		\begin{equation*}
			\|\tilde{\eta}\|\leq \Delta_{\tilde{\eta}}, \ \mbox{and} \ \|z_2\|\leq \Delta_{z_2},
		\end{equation*}
		where, $\Delta_{\tilde{\eta}}$ and $ \Delta_{z_2}$ are bounded finite positive constants that defines the size of neighbourhood $\mathcal{N}$. If the maximum norm bound of the disturbance torque is bounded by,
		\begin{equation}\label{normbound}
			\|\bar{\tau}_d\| \leq \Upsilon \leq \frac{\lambda_{m}(K_\eta)\Delta_{\tilde{\eta}}^{(1+\alpha)}+ \lambda_{m}(K_{z_2})\Delta_{z_2}^{(1+\alpha)}}{\Delta_{z_2}},
		\end{equation}
		then the closed loop feedback system \eqref{Kinamtics and dynamics with error2} converges to the neighbourhood $\mathcal{N}$.
	\end{corollary}
	\begin{proof}
		The proof follows from the Lyapunov stability analysis discussed in 
		\textit{Theorem} \ref{theo1}. Taking the time derivative of $\dot V(\tilde{\eta},z_2) =\dot V_1(\tilde{\eta})+\dot V_2(z_2)$ along the perturbed dynamics \eqref{Kinamtics and dynamics with error2} under the influence of control torque \eqref{actual controller} yields,
		\begin{equation*}
			\dot V(\tilde{\eta},z_2) =-\frac{\tilde{\eta}^{T}K_\eta\tilde{\eta}}{\left\|\tilde\eta\right\|^{(1-\alpha)}}-\frac{{z_2}^{T}K_zz_2}{\left\|{z_2}\right\|^{(1-\alpha)}}+z_2^T\bar{\tau}_d.
		\end{equation*}
		Compared to the unperturbed dynamical system, the extra term in the time derivative of $V(\tilde{\eta},z_2)$ arises as a consequence of the disturbance torque affecting the system. Taking the upper bound on the extra term containing the disturbance torque yields,
		\begin{equation*}
			z_2^T\bar{\tau}_d \leq \|z_2\| \|\bar{\tau}_d\| \leq \Delta_{z_2}\Upsilon.
		\end{equation*}
		Hence, the upper bound of $\dot V(\tilde{\eta},z_2)$ becomes,
		\begin{equation*}
			\dot V(\tilde{\eta},z_2) \leq-\frac{\tilde{\eta}^{T}K_\eta\tilde{\eta}}{\left\|\tilde\eta\right\|^{(1-\alpha)}}-\frac{{z_2}^{T}K_zz_2}{\left\|{z_2}\right\|^{(1-\alpha)}}+\Delta_{z_2}\Upsilon.
		\end{equation*}
		It can be further expressed as,
		\begin{equation*}
			\dot V(\tilde{\eta},z_2) \leq-\lambda_{m}(K_\eta)\|\tilde\eta\|^{(1+\alpha)}-\lambda_{m}(K_{z_2})\|z_2\|^{(1+\alpha)}+\Delta_{z_2}\Upsilon.
		\end{equation*}
		Consequently, evaluating the upper bound of $\dot V(\tilde{\eta},z_2)$ along the desired size of neighbourhood $\mathcal{N}$ yields,
		\begin{equation*}
			\dot V(\tilde{\eta},z_2) \leq-\lambda_{m}(K_\eta)\Delta_{\tilde{\eta}}^{(1+\alpha)}-\lambda_{m}(K_{z_2})\Delta_{z_2}^{(1+\alpha)}+\Delta_{z_2}\Upsilon.
		\end{equation*}
		Hence, along the boundary of $\mathcal{N}$, the time derivative of Lyapunov function $V(\tilde{\eta},z_2)$ will be negative definite if the condition,
		\begin{equation}\label{cond1}
			-\lambda_{m}(K_\eta)\Delta_{\tilde{\eta}}^{(1+\alpha)}-\lambda_{m}(K_{z_2})\Delta_{z_2}^{(1+\alpha)}+\Delta_{z_2}\Upsilon\leq 0,
		\end{equation}
		is satisfied, which is sufficient to ensure that the trajectories originating from outside $\mathcal{N}$ and in the domain of convergence of the equilibrium will eventually converge to $\mathcal{N}$. Further, equation \eqref{normbound} follows from \eqref{cond1} and hence the proof.
	\end{proof}

	\begin{remark}\label{rem1}
		It is important to note that, within a closed neighborhood of the equilibrium (origin), the desired size of $\mathcal{N}$, as characterized by $\Delta_{\tilde{\eta}}$ and $\Delta_{z_2}$, will be less than unity. According to \eqref{normbound}, in the finite-time scenario where $\alpha \in \left(\frac{1}{2},1\right)$, the maximum disturbance that can be tolerated is greater in comparison with the asymptotic case where $\alpha=1$. This highlights the improved robustness capabilities of the proposed finite-time controller. 
	In the next section, we carry out a case study including tracking the given mobile robot along a straight-line path with a humped obstacle placed along the path in order to illustrate this situation in our simulations. This impediment influences the system as an external disturbance. Therefore, in order for the system to go along the path, it must resist the extra disturbance torque.
%		To demonstrate this case in our simulations, we conducted a case study involving tracking of the given mobile robot along a straight-line path with a humped obstacle situated along the path. This obstacle acts as a constant external disturbance acting on the system. Additional disturbance torque as a result of this must be overcame by the system to move forward along the path.
	\end{remark}
	
	\section{Simulation, Case studies and Discussions}\label{sec:5}
	In this section, we present the real-time simulation which captures the hardware implementation of FWMR system. For this, we make use of Gazebo-ROS simulation platform. ROS is an open-source middleware employed in the development of robot software. 
	%This platform offers a range of services, including hardware abstraction, control of low-level devices, and libraries aimed at the creation of robotic applications. 
	Gazebo is an efficient real-time robotic simulator equipped with a complete toolbox of development libraries and cloud services to enable the simulation of physical designs in more realistic environments with high fidelity sensor streams. In this research, we have used Nexus 4WD Mecanum Simulator \citep{nexus1}: A ROS Repository for the robot description and gazebo simulation of the 4WD mecanum wheel robot.  It contains two packages:  \textit{Nexus 4wd mecanum description} (having URDF model and robot meshes) and $\textit{nexus\_4wd\_mecanuum\_gazebo}$ (having launch files and controller scripts written in python), which we have configured according to our objectives. The package $\textit{nexus\_4wd\_mecanuum\_description}$ contains the launch files which loads the nexus 4WD mecanum robot description parameter and the robot state publisher. Gazebo plugin for simulating the controller and launch files for running the simulation are contained in $\textit{nexus\_4wd\_mecanuum\_gazebo}$. Real-time simulation results are visualized using plot juggler: an application to plot logged data.  Gazebo-ROS architecture encompassing all these are outlined in Figure \ref{fig:ROS_architecture}.  
	\begin{figure}[htpb!]
		\centering
		\includegraphics[scale=0.27]{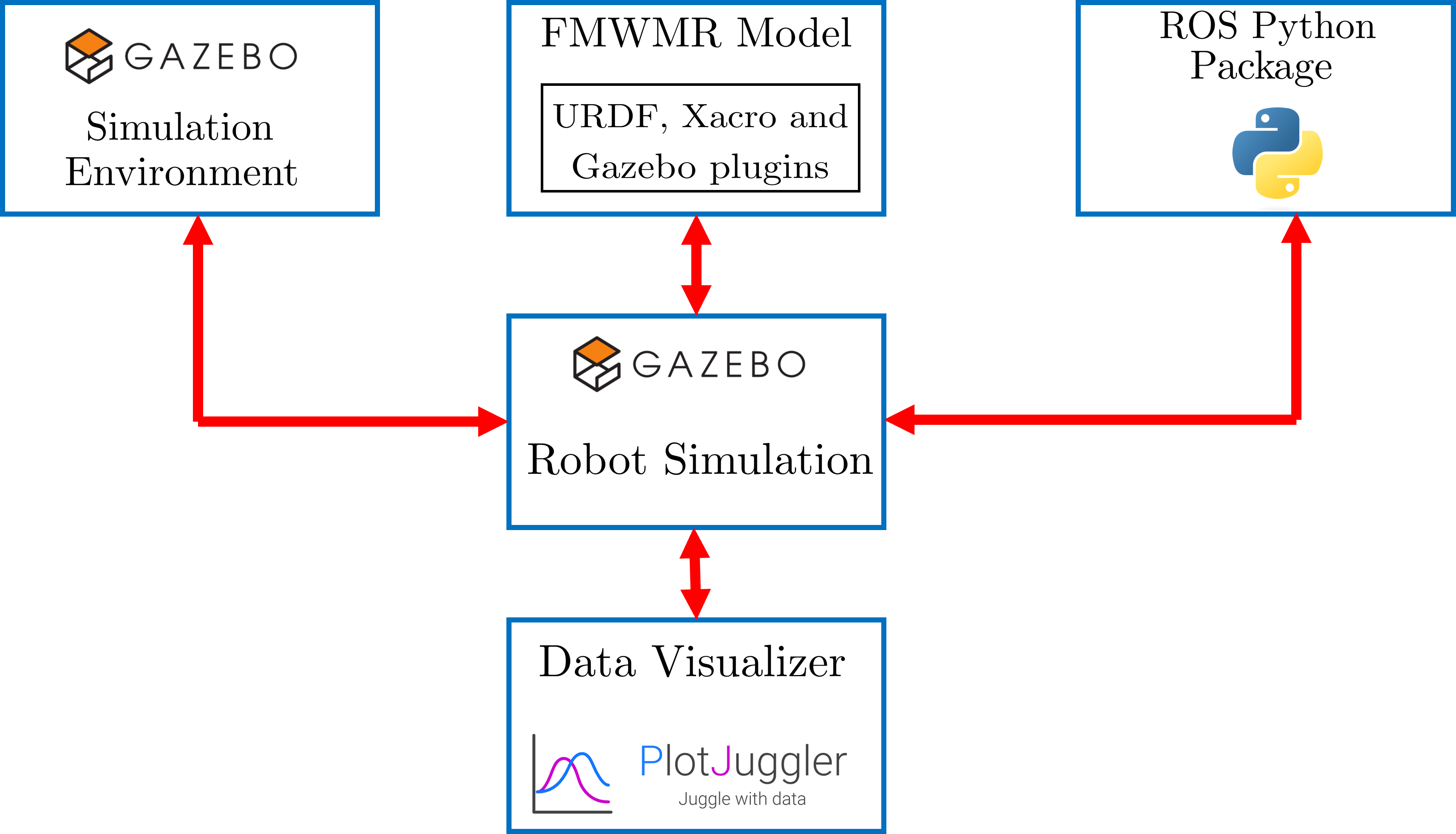}
		\caption{Gazebo-ROS architecture}
		\label{fig:ROS_architecture}
	\end{figure}
	Communication of messages in Gazebo-ROS through the ROS subscriber and publisher nodes are shown in ROS computation graph (Figure \ref{fig:Rqt}) generated through GUI plugin, $\textit{\mbox{rqt\_graph}}$.
	\begin{figure}[htpb!]
		\centering
		\includegraphics[scale=0.06]{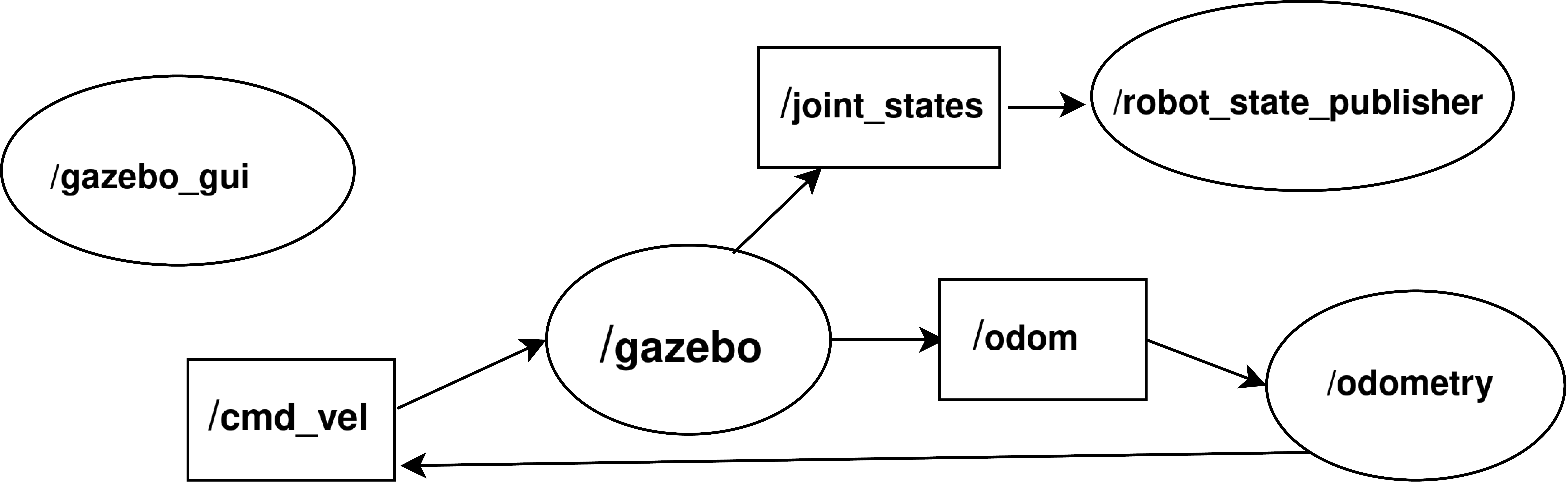}
		\caption{ROS computation graph}
		\label{fig:Rqt}
	\end{figure}
	\begin{figure}[htpb!]
		\centering
		\includegraphics[scale=0.3]{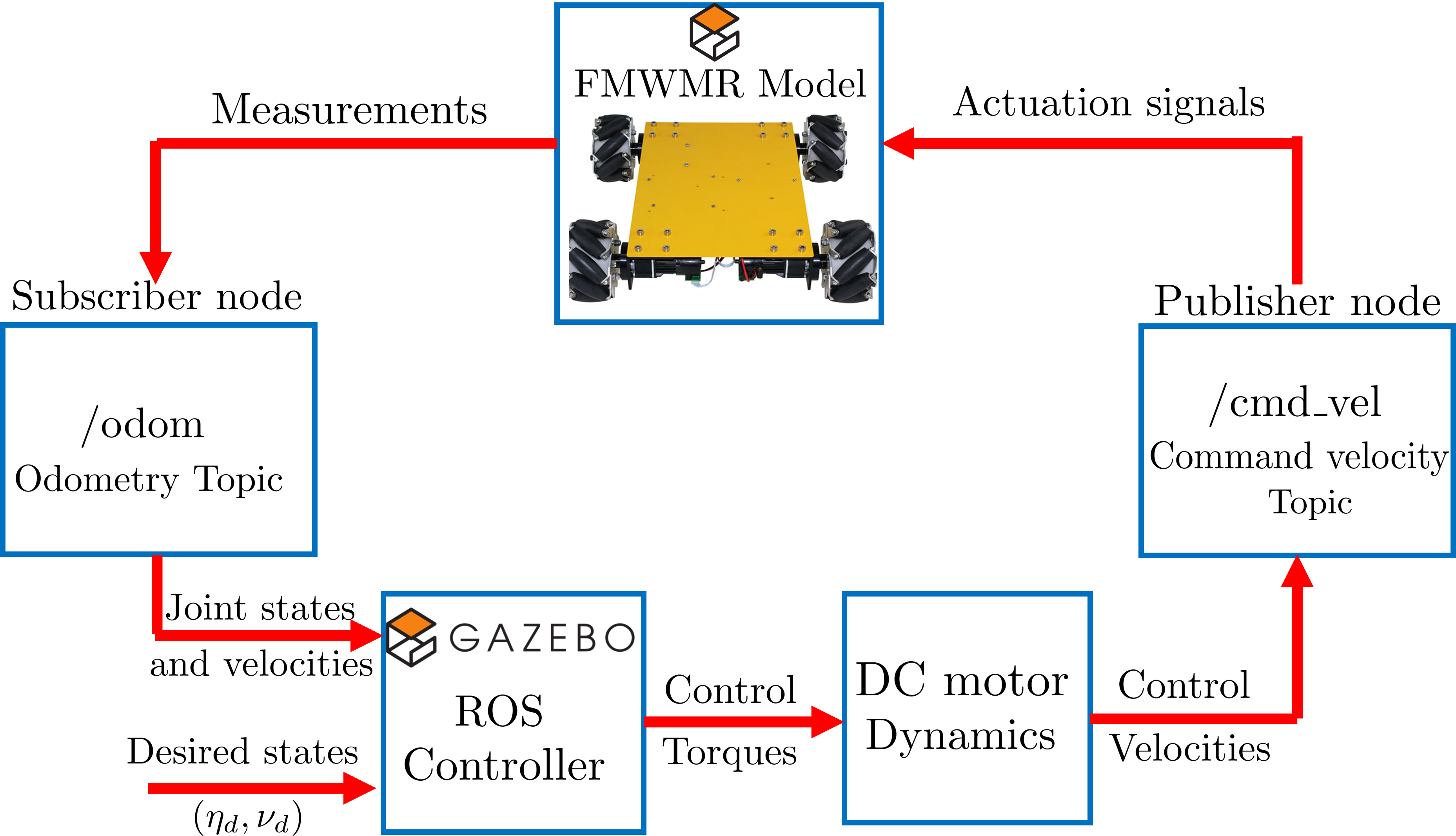}
		\caption{Closed loop control architecture of the system}
		\label{fig:cl}
	\end{figure}
	It has two nodes: a subscriber node, $/\mbox{odom}$ which subscribes to odometry messages from Gazebo plugins to gather vehicle pose and velocity informations and a publisher node, $/\mbox{cmd\_vel}$ which publishes velocity command (translational velocities along $x$ and $y$ directions and yaw rate about $z$-direction) as a twist message which forms the actuation signal to the robot model. The /gazebo\_gui topic is specifically associated with the graphical user interface (GUI) of Gazebo. This topic is  used to communicate information between ROS and the Gazebo GUI. It can be used for various purposes, such as sending commands to control the simulation, receiving sensor data from simulated robots, or updating the visualization of the simulation environment in the GUI. A closed loop control architecture with the designed controller is depicted in Figure \ref{fig:cl}. The subscriber node block receives the measurement signal (vehicle pose and velocities) from Gazebo mobile robot URDF model by subscribing to $/\mbox{odometry}$ topic.  These robot states and velocities along with the desired states are fed to the finite-time backstepping ROS controller block which computes the necessary control torques. These torques are then incorporated to update the position and velocity states in discrete-setting (choosing a sampling time of $0.01$s) approximated using Euler's backward difference scheme. The computed torque signals are transformed to required motor velocities using DC motor dynamics generating the motor control velocities which are then published through command velocity topic, $/\mbox{cmd\_vel}$. These command velocities form the actuation signal for the Gazebo URDF mobile robot model. This completes the closed loop control architecture. 
	\begin{remark}
		Gazebo-ROS provides built-in support for publishing velocities, making it straight forward to interface controller and evaluate it's performance within a ROS ecosystem. The command velocities will be converted into torques by means of a plugin for force-based move. It will generate output torque and forces to be applied to the base link of the robot, using  velocity commands as input. Our future work will involve implementing the developed control strategy on a physical FWMR system where, direct torque commands will be applied for real-world validation.
	\end{remark}
	Now, we present the Gazebo-ROS simulation results with the developed  finite-time backstepping controller \eqref{actual controller}. We present the comparative assessment of finite-time and asymptotic controllers for achieving stabilization and tracking control objectives. Asymptotic feedback law can be obtained by setting $\alpha$=1 in \eqref{actual controller}. Two scenarios are presented: a point stabilization problem and a tracking control problem. 
	\subsection{Case 1: Point stabilization}
	Here, we consider the problem of stabilization of system states and velocities to origin from an arbitrary initial condition. Though this is a sub case of tracking it is not trivial  since we carry out this objective in our simulation for demonstrating few analytical results on the effect of tuning parameters such as control gains and finite-time parameter $\alpha$ on system performance. Initial errors for system states is chosen as: $(\tilde{\eta}(0),z_2(0))=(5,-4,\pi/4,1,0.5,-0.5)$. For these initial conditions, theoretically calculated guaranteed finite-time from \eqref{time} is around $6.4031$s. Figures \ref{fig:POE_stab} to  \ref{fig:POE_stab2} shows the position and orientation stabilization errors. The practical value of average finite-time calculated for pose stabilization ($\vert\tilde{\eta}\vert\leq10^{-6}$) is about $4.24$s for finite-time case. For  asymptotic case, the average convergence time for the same was found to be around $10.03$s. The error in velocities, $z_2$ for subsystem 2 is shown in Figures \ref{fig:VE_stab1} to \ref{fig:VE_stab3}.  Practical value of average convergence time for $\vert z_2\vert\leq10^{-6}$ is obtained around $4.51$s while  for the same convergence criteria, it is obtained around $12.62$s for the asymptotic case. This clearly indicates the faster convergence capabilities of the proposed finite-time strategy in contrast with the conventional asymptotic feedback control law. Figure \ref{fig:u_stab} shows the control inputs (forces and torque) applied to the base link of the robot which becomes zero at steady state once the states converges to the desired values.  
	\begin{table}[h!tbp]
		\centering
		\begin{tabular}{|c|c|c|c|c|c|}
			\hline
			\multirow{2}{*}{Case 1} &  \multicolumn{2}{c|}{Settling Time} & \multicolumn{3}{c|}{TV} \\
			\cline{2-6}
			& $\tilde{\eta}$ & $z_2$ & $\bar{\tau}(1)$ & $\bar{\tau}(2)$ & $\bar{\tau}(3)$ \\
			\hline
			Finite-time & 4.24s & 4.51s & 56.11 & 55.98 & 55.83\\
			Asymptotic & 10.03s& 12.62s & 116.76 & 115.21  & 113.62\\
			\hline
		\end{tabular}
		\caption{Quantitative comparison: Case 1}
		\label{tab:mytable1}
	\end{table}
	A quantitative comparison in terms of convergence time and control effort is provided in Table \ref{tab:mytable1}. We have used Total Variation (TV$=\sum_{k=2}^{N}\vert\bar{\tau}_{k}(i)-\bar{\tau}_{k-1}(i)\vert$, where $i=1$ to $3$ and $N$ is the number of samples) as a measure for quantitative evaluation of control effort. This is indicative of control signal fluctuations which should be as low as possible to prevent actuator wear and tear. TV measure is found to be low for finite-time case compared to the asymptotic case thereby indicating the efficiency of our proposed controller in terms of actuation. 
		\begin{figure}[htpb!]
		\centering
		\includegraphics[scale=0.5]{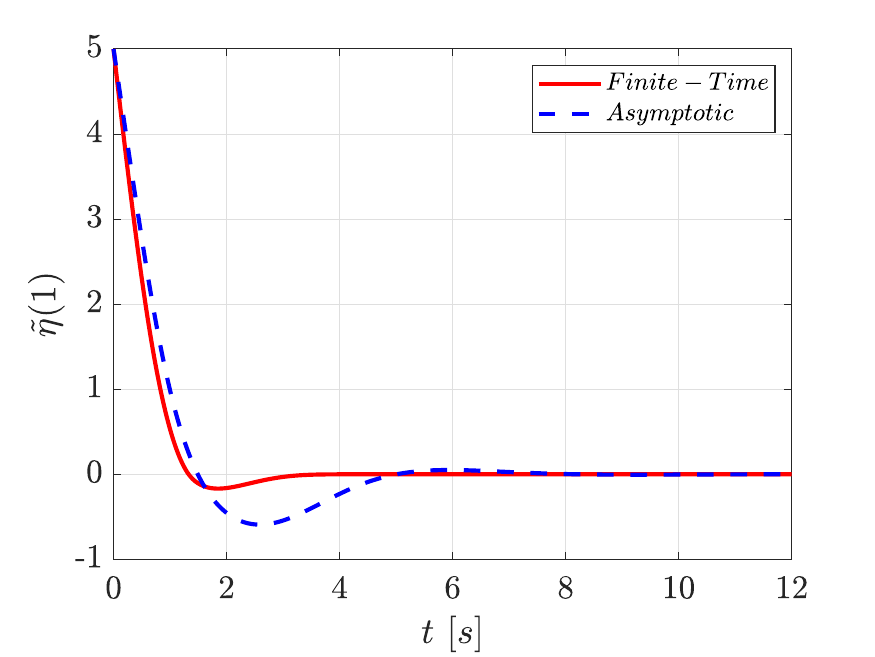}
		\caption{Case 1: Pose error $\tilde{\eta} (1)$ corresponding to $x$-position with time; convergence time for $\vert \tilde{\eta}(1) \vert \leq10^{-6}$ is obtained around 4.35s and 10.2s for finite-time and asymptotic cases respectively.}
		\label{fig:POE_stab}
	\end{figure}
	\begin{figure}[htpb!]
		\centering
		\includegraphics[scale=0.5]{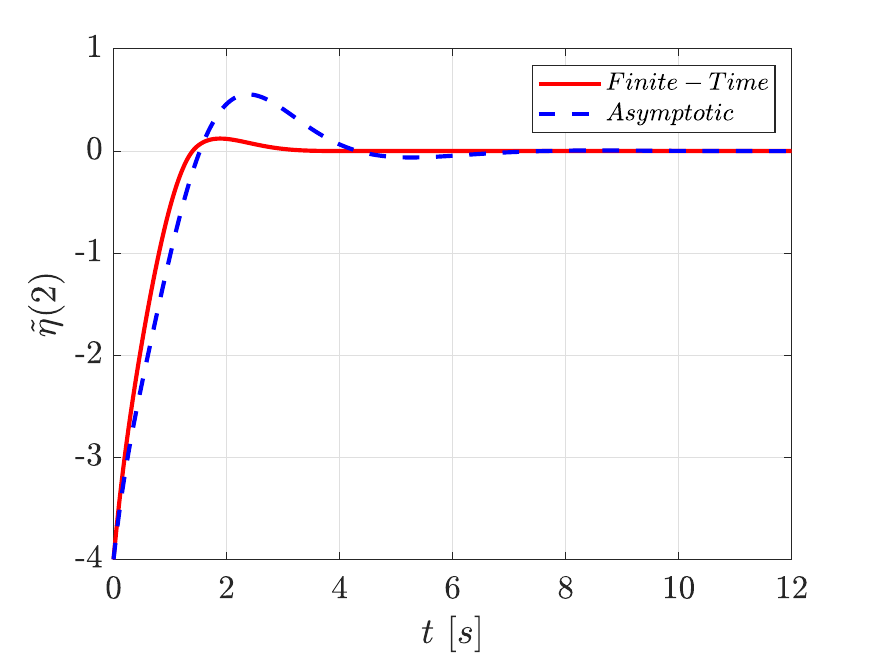}
		\caption{Case 1: Pose error $\tilde{\eta} (2)$ corresponding to $y$-position with time; convergence time for $\vert \tilde{\eta}(2) \vert \leq10^{-6}$ is obtained around 4.24s and 10.1s for finite-time and asymptotic cases respectively.}
		\label{fig:POE_stab1}
	\end{figure}
	\begin{figure}[htpb!]
		\centering
		\includegraphics[scale=0.5]{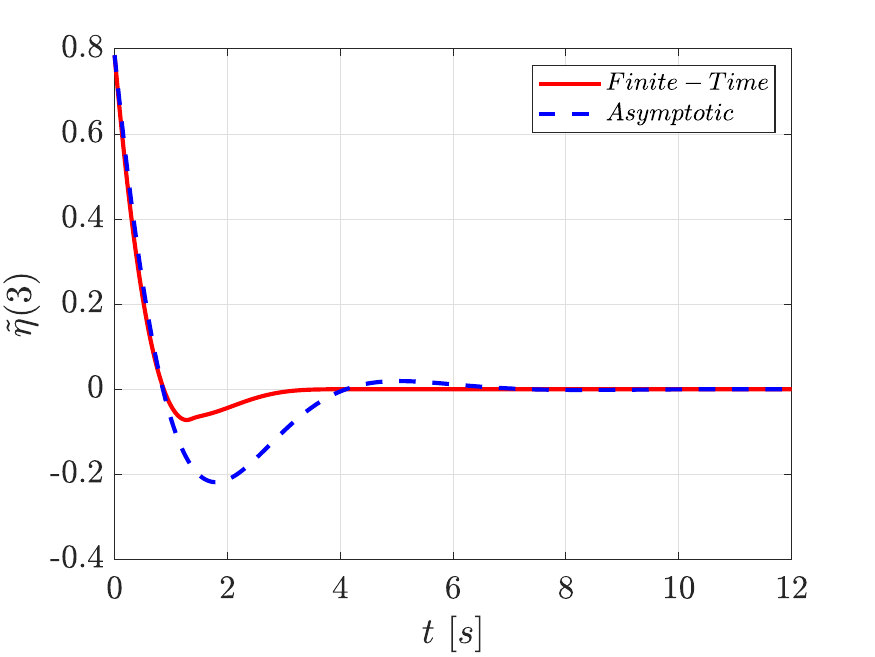}
		\caption{Case 1: Pose error $\tilde{\eta} (3)$ corresponding to yaw angle with time; convergence time for $\vert \tilde{\eta}(3) \vert \leq10^{-6}$ is obtained around 4.12s and 9.8s for finite-time and asymptotic cases respectively.}
		\label{fig:POE_stab2}
	\end{figure}
	\begin{figure}[htpb!]
		\centering
		\includegraphics[scale=0.5]{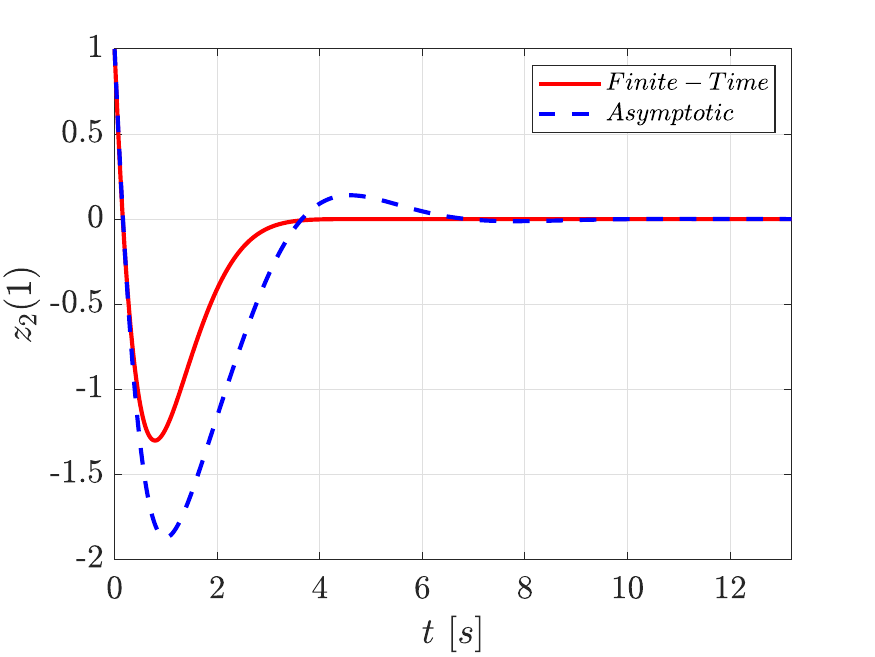}
		\caption{Case 1: Velocity error $z_2(1)$ corresponding to heading velocity with time; convergence time for $\vert z_2(1) \vert \leq10^{-6}$ is obtained around 4.56s and 12.8s for finite-time and asymptotic cases respectively.}
		\label{fig:VE_stab1}
	\end{figure}
	\begin{figure}[htpb!]
		\centering
		\includegraphics[scale=0.5]{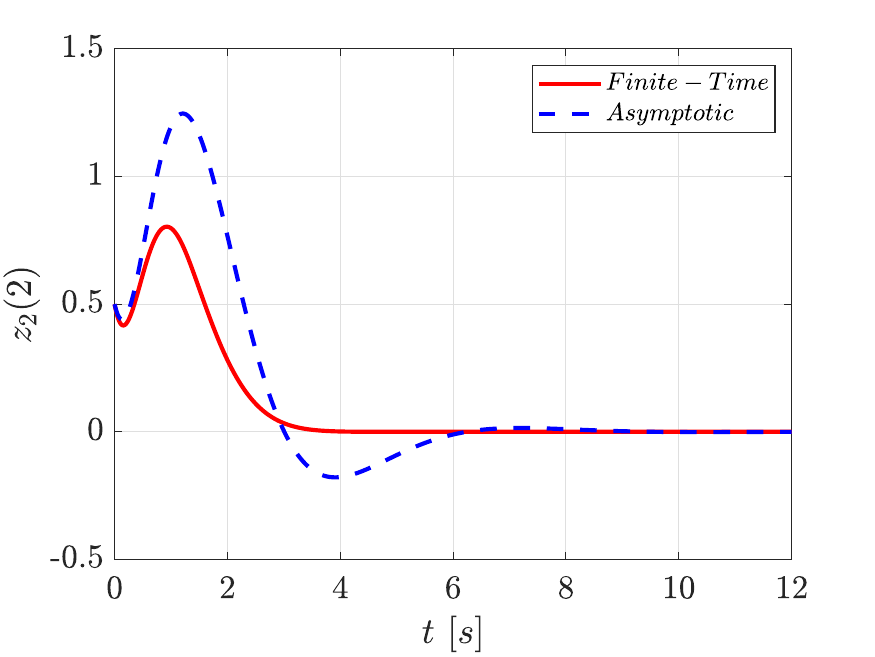}
		\caption{Case 1: Velocity error $z_2(2)$ corresponding to lateral velocity  with time; convergence time for $\vert z_2(2) \vert \leq10^{-6}$ is obtained around 4.6s and 12.65s for finite-time and asymptotic cases respectively.}
		\label{fig:VE_stab2}
	\end{figure}
	\begin{figure}[htpb!]
		\centering
		\includegraphics[scale=0.5]{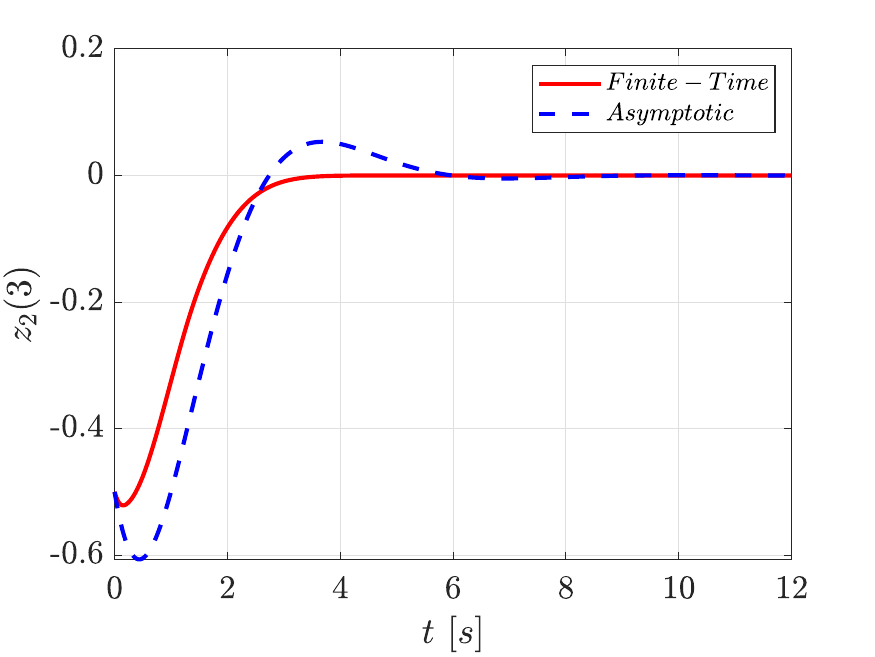}
		\caption{Case 1: Velocity error $z_2(3)$ corresponding to yaw rate with time; convergence time for $\vert z_2(3) \vert \leq10^{-6}$ is obtained around 4.37s and 12.4s for finite-time and asymptotic cases respectively.}
		\label{fig:VE_stab3}
	\end{figure}
	\begin{figure}[htpb!]
		\centering
		\includegraphics[scale=0.55]{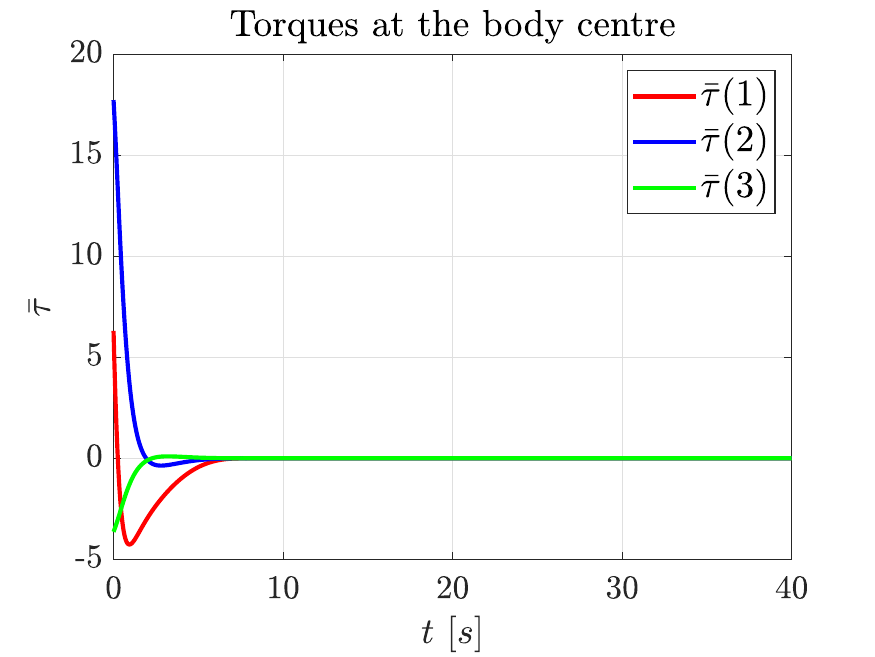}
		\caption{Case 1: Control forces and torque}
		\label{fig:u_stab}
	\end{figure}
	
	\subsubsection{Guidelines for tuning the controller weights $K_\eta$ and $K_{z}$}
	We present the analytical results to examine how varying the tuning parameters $K_\eta$ and $K_{z}$ impacts the system performance measures $\mathcal{T}$ (the finite-time) and TV. Specifically, we characterize the tuning gains $K_\eta$ and $K_{z}$ in terms of the parameters $\lambda_1$ and $\lambda_2$ as follows: $K_\eta = \text{diag}(\lambda_1, \lambda_1, 0.5\lambda_1)$ and $K_z = \text{diag}(\lambda_2, \lambda_2, 0.5\lambda_2)$. Figure \ref{fig:3dmesh} illustrates a surface plot showing how $\mathcal{T}$ and TV vary with changes in $\lambda_1$ and $\lambda_2$. The plot reveals that increasing the tuning weights results in a decrease in convergence time $\mathcal{T}$, but leads to a rise in TV due to large deviation in $\dot{\Psi}(\tilde{\eta})$. This indicates a trade-off between these two performance metrics as tuning weights are adjusted. Figure \ref{tradeoff} further highlights this trade-off. Therefore, the selection of tuning weights should be made with careful consideration of this balance.
		\begin{figure}[htpb!]
		\centering
		\includegraphics[scale=0.33]{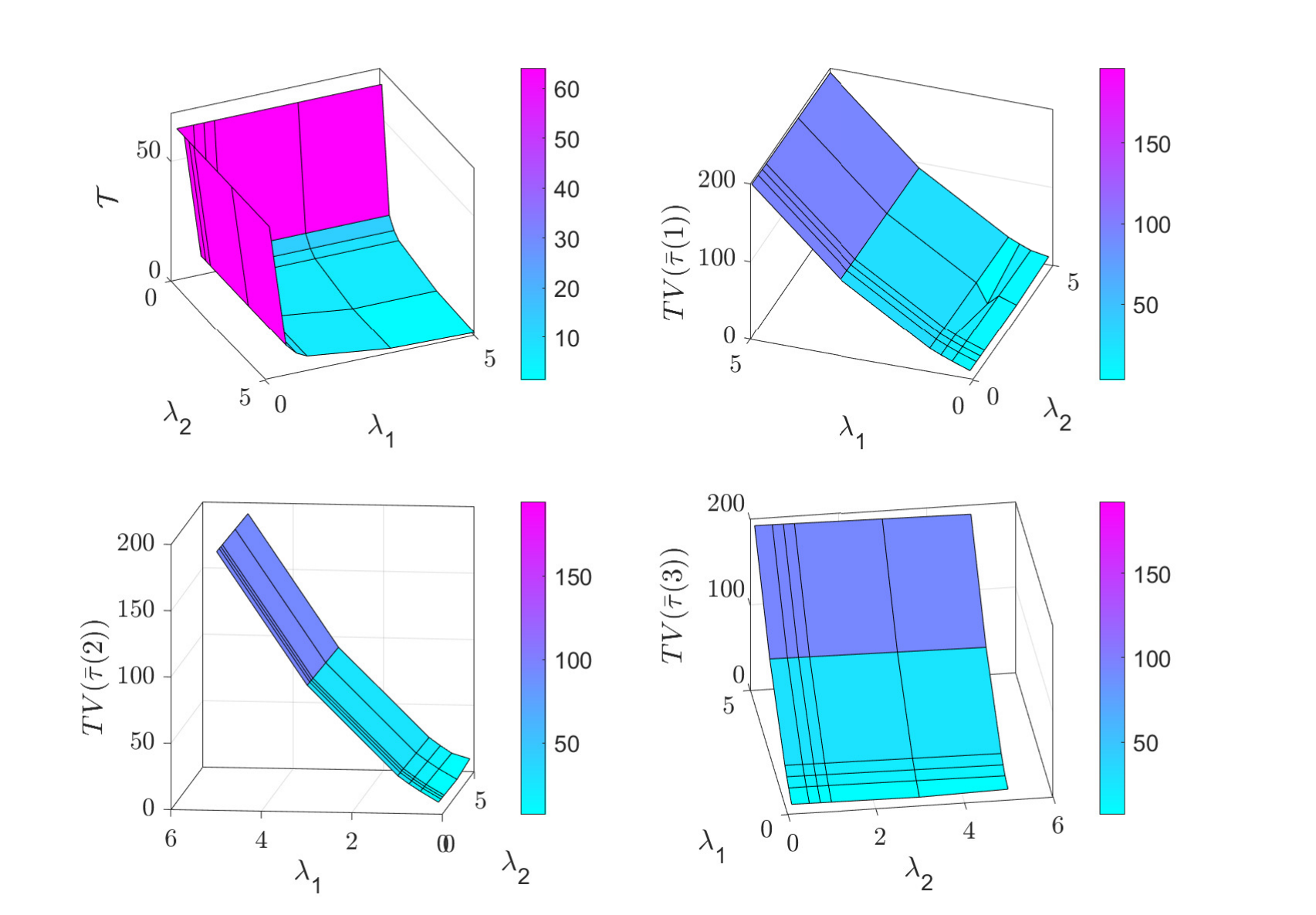}
		\caption{ Variation of $\mathcal{T}$ and TV with the tuning parameters $\lambda_1$ and $\lambda_2$; Here, $\lambda_1$ and $\lambda_2$ are chosen as common parameters which characterizes the weights $K_\eta$ and $K_{z}$ as: $K_\eta = \text{diag}(\lambda_1, \lambda_1, 0.5\lambda_1)$ and $K_z = \text{diag}(\lambda_2, \lambda_2, 0.5\lambda_2)$}
		\label{fig:3dmesh}
	\end{figure}
	\begin{figure}[htpb!]
		\centering                                                                  	
		\includegraphics[scale=0.55]{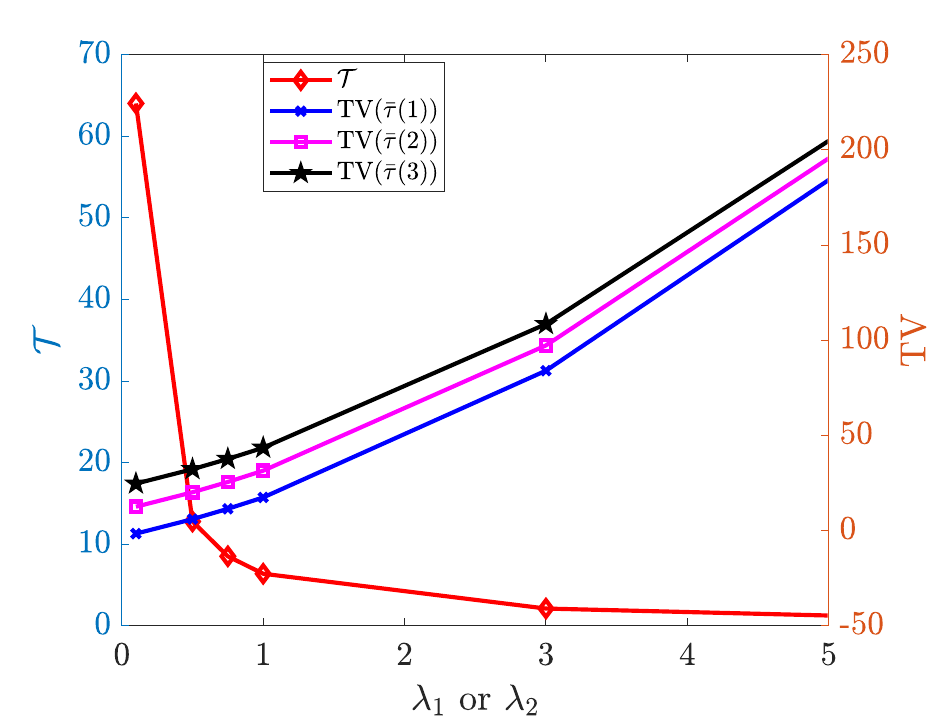}
		\caption{Factors affecting the choice of controller gains $K_{\eta}$ and $K_{z}$; It clearly shows the trade-off between the performance metrics $\mathcal{T}$ and TV as the tuning weights are adjusted}
		\label{tradeoff}
	\end{figure}
	\subsubsection{Factors governing the choice of finite-time parameter $\alpha$}
	The parameter $\alpha$ has a significant impact on finite time $\mathcal{T}$, total variation (TV), and robustness to external disturbances. To illustrate this analysis, we perform various comparisons for case 1, which focuses on point stabilization. However, this approach is applicable to any scenario, whether it's tracking or stabilization. Figure \ref{fig:alphaTV} illustrates how TV varies with $\alpha$ over the range from 0.5 to 1. 
	\begin{figure}[h!tpb]
		\centering
		\includegraphics[scale=0.5]{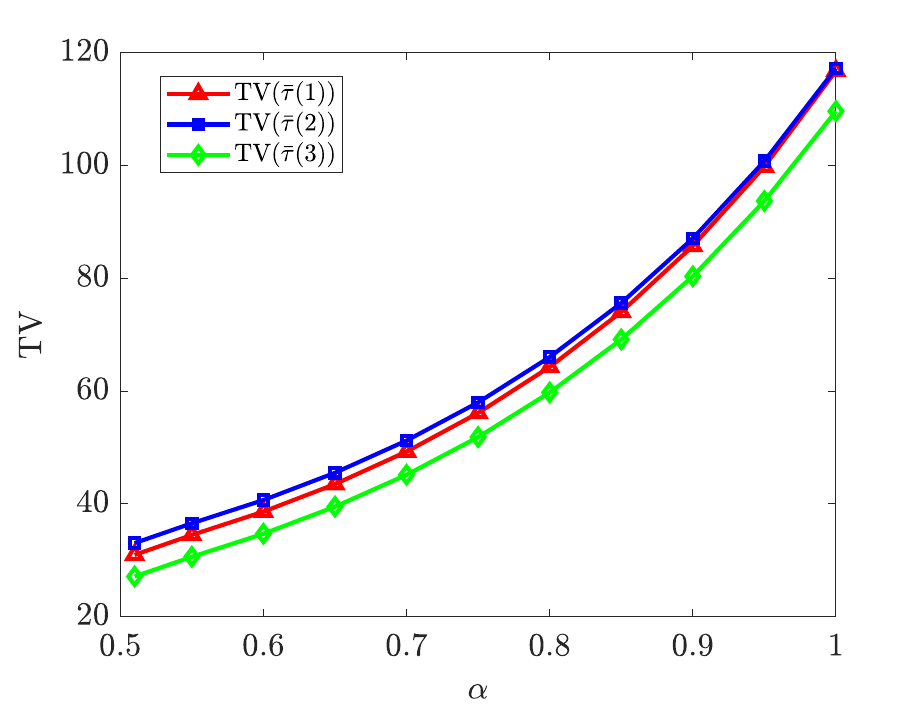}
		\caption{Plot showing the effect of $\alpha$ on TV; Illustrates how TV varies with $\alpha$ over the range from 0.5 to 1}
		\label{fig:alphaTV}
	\end{figure}
	\begin{figure}[h!tpb]
	\centering
	\includegraphics[scale=0.5]{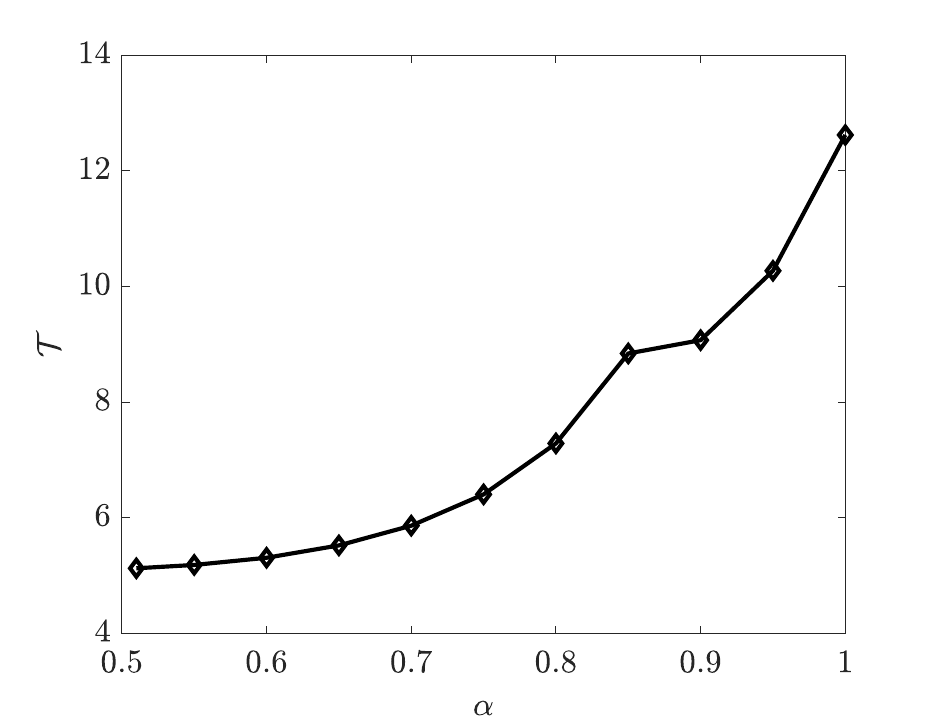}
	\caption{Plot showing the effect of $\alpha$ on $\mathcal{T}$}
	\label{fig:alphaT}
\end{figure}
	Figure \ref{fig:alphaT} shows the relationship between $\mathcal{T}$ and $\alpha$. The increasing trend in the graph indicates that, compared to the asymptotic case where $\alpha=1$, a finite-time controller offers faster convergence. The graph also demonstrates that TV increases with $\alpha$, but it is notably lower for finite-time controllers compared to the asymptotic case. This suggests that the finite-time controller provides better actuation capabilities.	
Additionally, Figure \ref{fig:robust2} presents the trend of TV and the maximum normalized external disturbance torque bound with varying $\alpha$. Lower values of $\alpha$ enhance robustness while reducing the controller effort. Thus, the proposed finite-time controller offer substantial advantages in these aspects. However, in practice, values of $\alpha$ close to 0.5 may cause flickering in the control torque as elucidated by \textit{Lemma} \ref{lemma1} which will reflect as chattering in the velocity plots. Therefore, the choice of $\alpha$ should balance these factors. Based on several analytical simulations, a recommended range for $\alpha$ is between 0.6 to 0.8.

	\begin{figure}[h!tpb]
		\centering
		\includegraphics[scale=0.57]{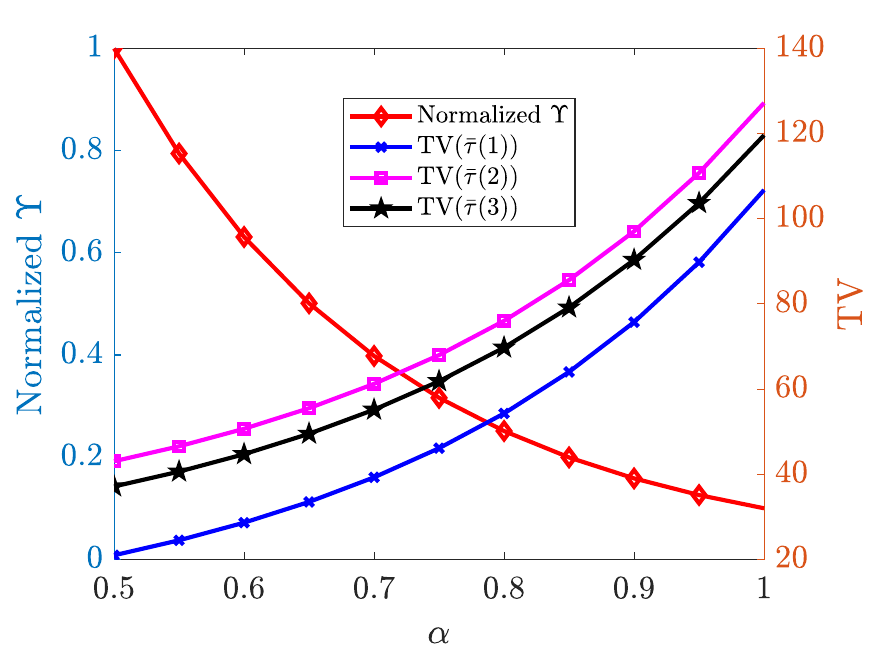}
		\caption{Plot showing the trend of TV and the maximum normalized disturbance torque bound upon varying $\alpha$}
		\label{fig:robust2}
	\end{figure}
	 
\begin{figure*}[h!tpb]
	\centering
	\includegraphics[scale=0.125]{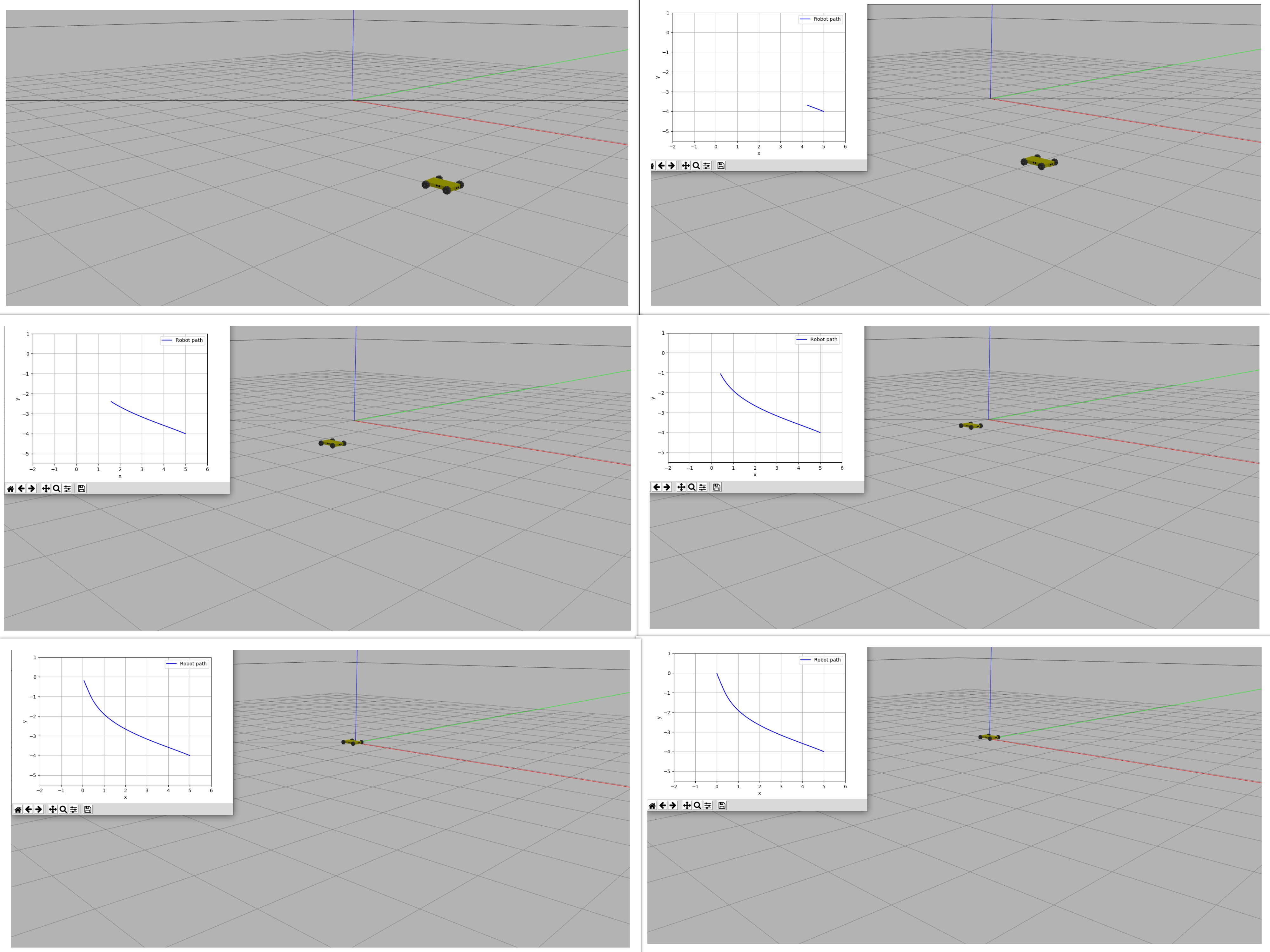}
	\caption{Few snapshots illustrating the stabilization objective of FWMR in Gazebo-ROS considering some initial condition for the system states; Sequence of images progresses from the top left to the bottom right, capturing key moments of the vehicle path.}
	\label{stsc}
\end{figure*}

	%\begin{figure}[htpb!]
	%	\centering
	%	\includegraphics[scale=0.6]{phase_eps_t}
	%	\caption{Case 1: Stabilization plot in $XY$ plane}
	%	\label{fig:phase_stab}
	%\end{figure}
	\subsection{Case 2: S-shaped trajectory tracking}
This case is motivated by motion planning application, to avoid collision with the walls by passing through narrow passages in warehouses as depicted in Figure \ref{fig:phase_track}. Here, the robot has to navigate and execute a zig-zag motion which can be parameterized by an S-shaped trajectory.  This forms a trajectory tracking control problem. Initial errors for system states and velocities,  $(\tilde{\eta}(0),z_2(0))$=(-1,-2,$\pi$/4,-1,2,1) are chosen for simulation. The theoretically calculated guaranteed finite-time for these choice of initial conditions is approximately $5.2357$ seconds. Figures \ref{fig:POE_stabc} to \ref{fig:POE_stabc2} illustrates the errors in position and orientation tracking. For the finite-time case, the practical average convergence time for pose tracking ($\vert\tilde{\eta}\vert\leq10^{-6}$) is approximately $4.38$ seconds. In contrast, for the asymptotic case, the average convergence time was found to be around $14.5$ seconds. The error in velocities, $z_2$ for subsystem 2 is depicted in Figures \ref{fig:VE_stabc1} to \ref{fig:VE_stabc3}. The practical convergence time for $\vert z_2\vert\leq10^{-6}$ is approximately $4.58$ seconds for the finite-time case and $12.54$ seconds for the asymptotic case. This highlights the superior convergence capabilities of proposed finite-time strategy. 
	Figure \ref{fig:u_track} illustrates the actual control inputs (forces and torque) applied to the body center of the robot, which is non-zero at steady state since it is a tracking control problem. The S-shaped trajectory tracking plot in the $XY$ plane for finite-time case along with the desired trajectory is shown in Figure \ref{fig:phase_track}. 
		\begin{table}
		\begin{tabular}{|c|c|c|c|c|c|}
			\hline
			\multirow{2}{*}{Case 2} &  \multicolumn{2}{c|}{Settling Time} & \multicolumn{3}{c|}{TV} \\
			\cline{2-6}
			& $\tilde{\eta}$ & $z_2$ & $\bar{\tau}(1)$ & $\bar{\tau}(2)$ & $\bar{\tau}(3)$ \\
			\hline
			Finite-time & 4.38s & 4.58s & 11.45 & 14.19 & 2.37\\
			Asymptotic & 14.50s& 12.54s& 30.16 & 36.19  & 11.90\\
			\hline
		\end{tabular}
		\caption{Quantitative comparison: Case 2}
		\label{tab:mytable2}
	\end{table}
	A quantitative comparison in terms of convergence time and control effort is provided in Table \ref{tab:mytable2} where, it can be inferred that the proposed finite-time backstepping based feedback control law yields lesser settling time as well as control signal fluctuations compared with the conventional asymptotic feedback control scheme. Real-time simulation results for stabilization and tracking we have performed in Gazebo-ROS is available for reference  \href{https://docs.google.com/presentation/d/1cuLrpexPpeugGZt6EsnN4Xlp_H-SjbU2rH4_Fab523I/edit?usp=sharing}{in this link}. We have also provided few snapshots of tracking the trajectory in Gazebo-ROS which is shown in Figure \ref{Sgazebo}.
	
	\subsection{Case 3: Robustness to external perturbations}
	In this case, we demonstrate the robustness capabilities of the proposed finite-time controller subjected to external disturbances, We conducted a case study involving the tracking of a mobile robot along a straight-line path, with a humped obstacle positioned along the route to simulate this scenario in our simulations. This obstacle acts as an external disturbance, necessitating that the system compensates for the additional torque in order to maintain its trajectory along the path. This scenario is depicted in Figure \ref{obs}.
In this scenario, the goal is to travel along a straight line on the \( x \)-axis while encountering an obstacle. Consequently, the vehicle must provide additional torque to surmount the obstacle and continue moving forward. This requirement arises from compensating for external disturbances, which are consistently present. The maximum tolerable level of these disturbances depends on both transient and control weights. Plot of $x$ versus $z$ position is depicted in Figure \ref{XZ}, where the path represents geometry of the obstacle situated along the path. Figure \ref{torqueG} illustrates the control input plot where an increase in control effort is attributed to the obstacle in its path, which necessitates additional input to overcome. Snapshots of the robot maneuver in the presence of an obstacle along it's path is depicted in Figure \ref{Obs}. This result validates remark \ref{rem1}.
A video showcasing the results of Gazebo-ROS simulations conducted for all the three cases is available at the following link: \url{https://youtu.be/F6paTC0Xp3c}.

	\begin{figure}[htpb!]
	\centering
	\includegraphics[scale=0.25]{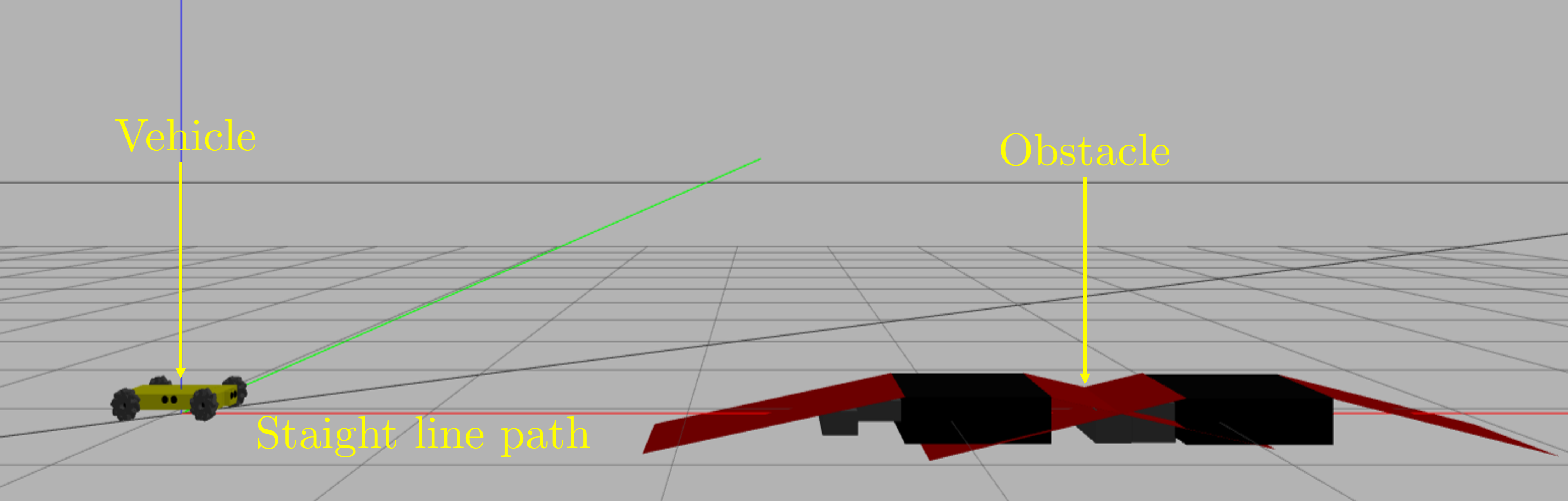}
	\caption{Case 3: Depicts the scenario where vehicle is subjected to external disturbances induced by an obstacle along it's path.}
	\label{obs}
\end{figure}
	\begin{figure}[htpb!]
		\centering
		\includegraphics[scale=0.5]{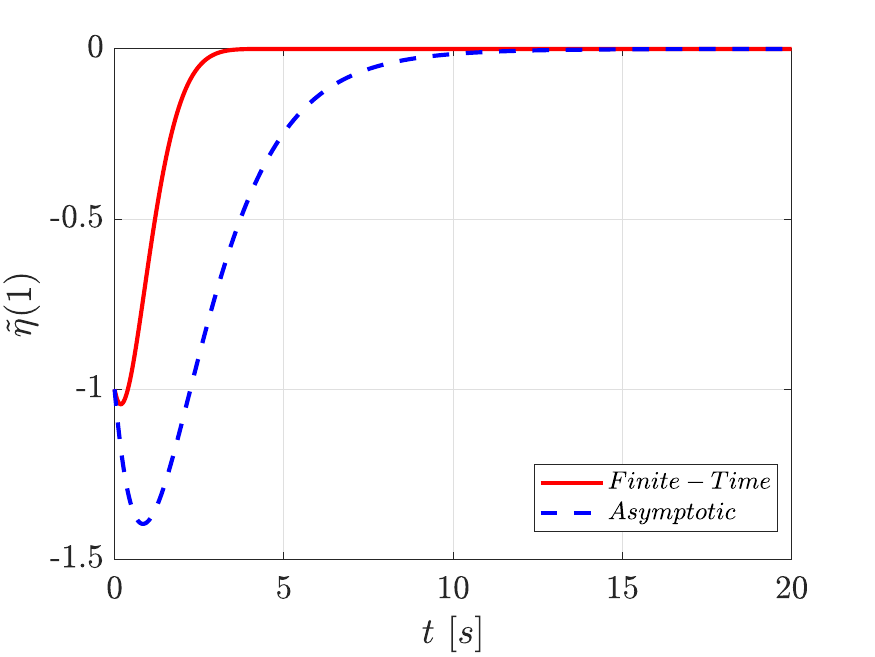}
		\caption{Case 2: Pose tracking error $\tilde{\eta} (1)$ corresponding to $x$-position with time; convergence time for $\vert \tilde{\eta}(1) \vert \leq10^{-6}$ is obtained around 4.42s and 14.45s for finite-time and asymptotic cases respectively.}
		\label{fig:POE_stabc}
	\end{figure}
	\begin{figure}[htpb!]
		\centering
		\includegraphics[scale=0.5]{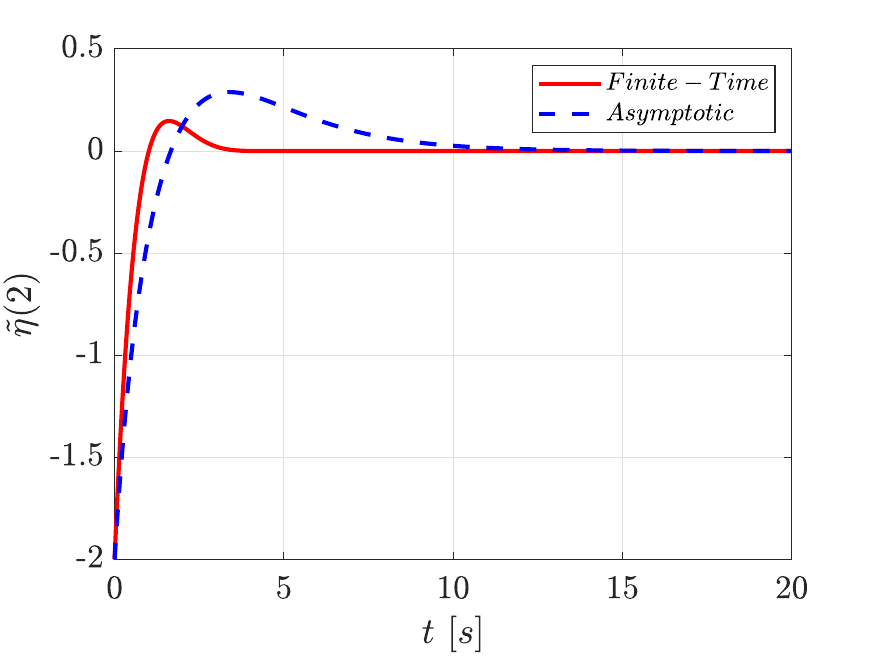}
		\caption{Case 2: Pose tracking error $\tilde{\eta} (2)$ corresponding to $y$-position with time; convergence time for $\vert \tilde{\eta}(2) \vert \leq10^{-6}$ is obtained around 4.56s and 14.84s for finite-time and asymptotic cases respectively.}
		\label{fig:POE_stabc1}
	\end{figure}
	\begin{figure}[htpb!]
		\centering
		\includegraphics[scale=0.5]{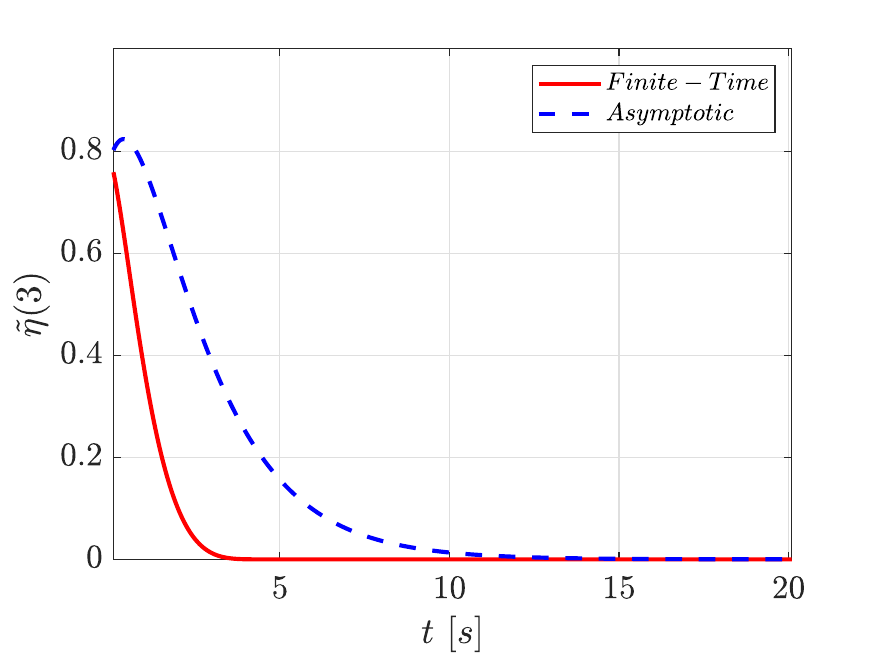}
		\caption{Case 2: Pose error tracking $\tilde{\eta} (3)$ corresponding to yaw angle with time; convergence time for $\vert \tilde{\eta}(3) \vert \leq10^{-6}$ is obtained around 4.17s and 14.2s for finite-time and asymptotic cases respectively.}
		\label{fig:POE_stabc2}
	\end{figure}
	\begin{figure}[htpb!]
		\centering
		\includegraphics[scale=0.5]{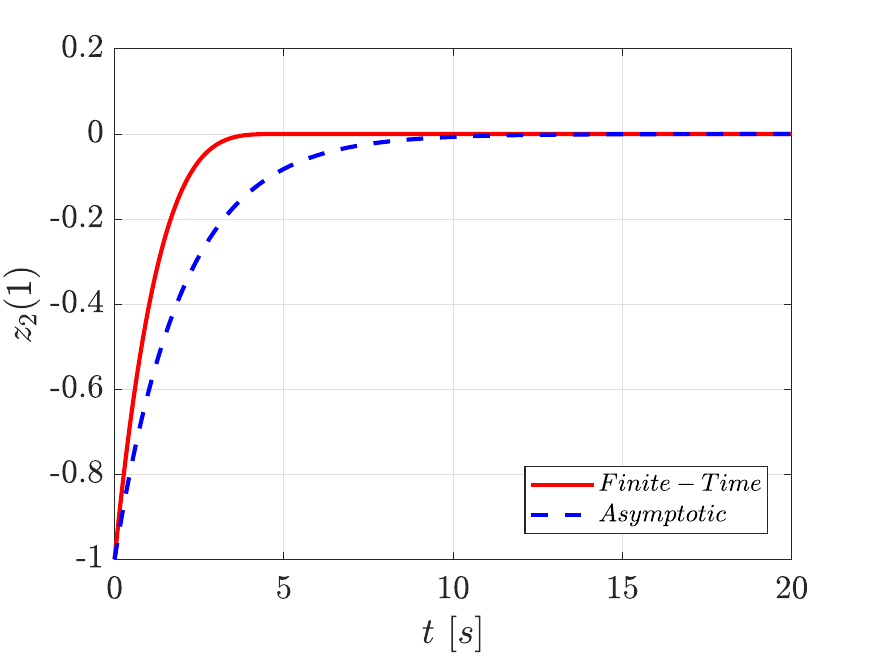}
		\caption{Case 2: Velocity tracking error $z_2(1)$ corresponding to heading velocity with time; convergence time for $\vert z_2(1) \vert \leq10^{-6}$ is obtained around 4.85s and 13.24s for finite-time and asymptotic cases respectively.}
		\label{fig:VE_stabc1}
	\end{figure}
	\begin{figure}[htpb!]
		\centering
		\includegraphics[scale=0.5]{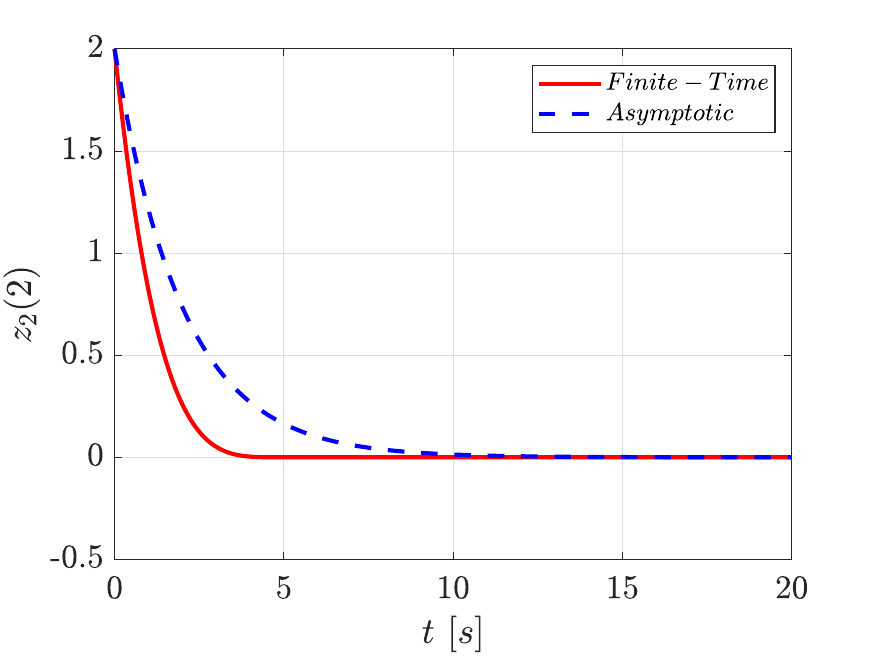}
		\caption{Case 2: Velocity tracking error $z_2(2)$ corresponding to lateral velocity  with time; convergence time for $\vert z_2(2) \vert \leq10^{-6}$ is obtained around 4.73s and 12.85s for finite-time and asymptotic cases respectively.}
		\label{fig:VE_stabc2}
	\end{figure}
	\begin{figure}[htpb!]
		\centering
		\includegraphics[scale=0.5]{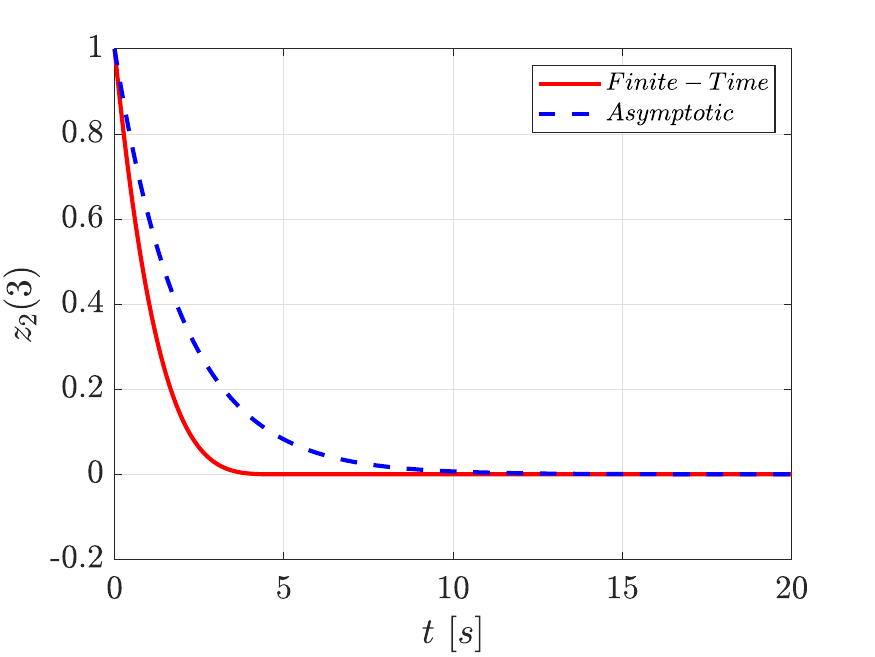}
		\caption{Case 2: Velocity tracking error $z_2(3)$ corresponding to yaw rate with time; convergence time for $\vert z_2(3) \vert \leq10^{-6}$ is obtained around 4.15s and 11.54s for finite-time and asymptotic cases respectively.}
		\label{fig:VE_stabc3}
	\end{figure}

	\begin{figure}[htpb!]
		\centering
		\includegraphics[scale=0.55]{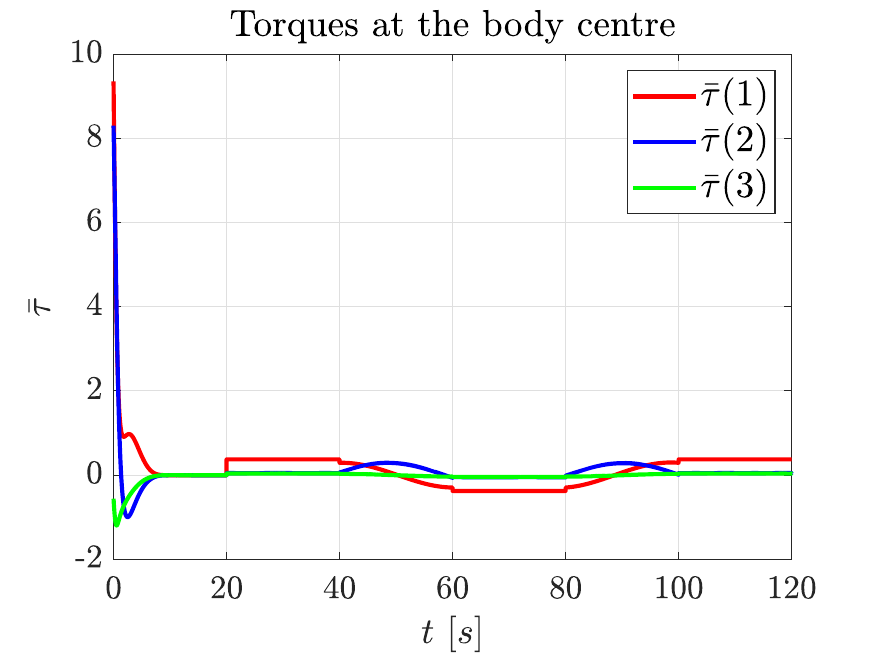}
		\caption{Case 2: Control forces and torque}
		\label{fig:u_track}
	\end{figure}
	\begin{figure}[htpb!]
		\centering
		\includegraphics[scale=0.55]{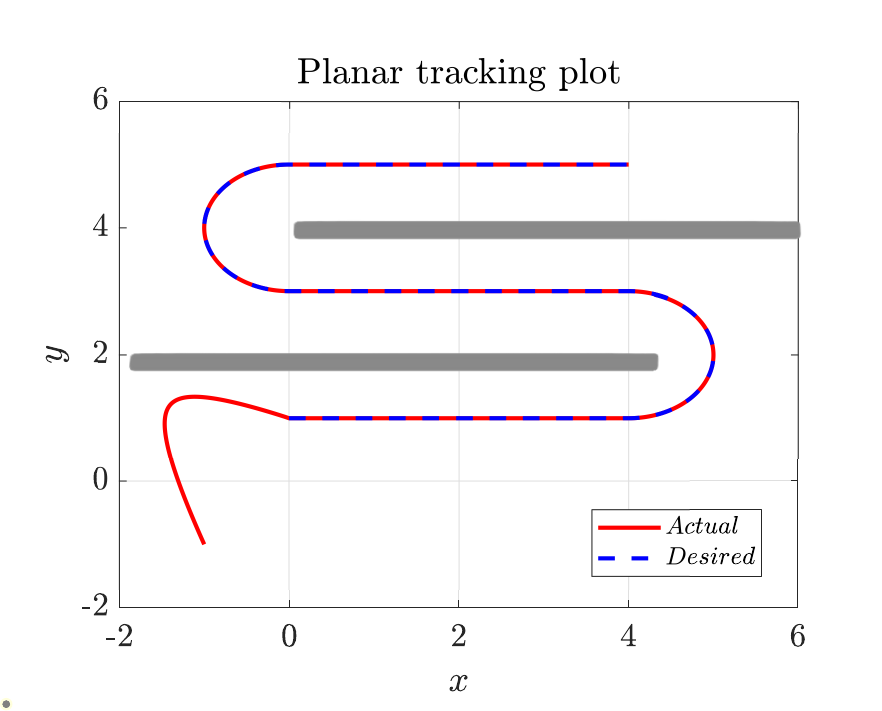}
		\caption{Case 2: Finite-time tracking in $XY$ plane}
		\label{fig:phase_track}
	\end{figure}

\begin{figure}[htpb!]
	\centering
	\includegraphics[scale=0.55]{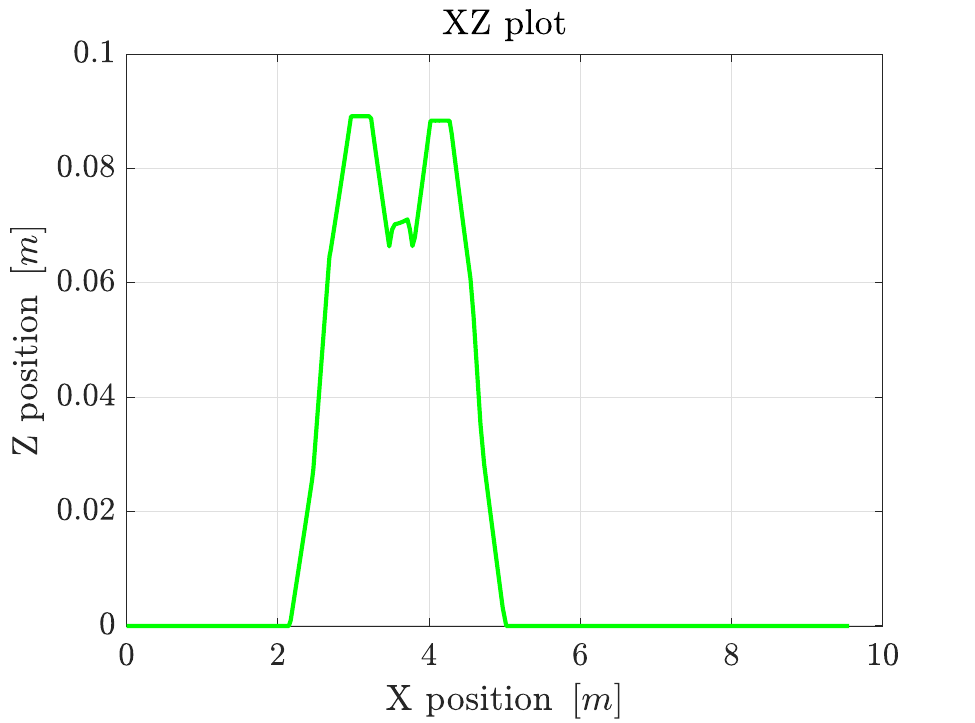}
	\caption{XZ plot depicting the path of vehicle over the obstacle}
	\label{XZ}
\end{figure}
\begin{figure}[htpb!]
	\centering
	\includegraphics[scale=0.55]{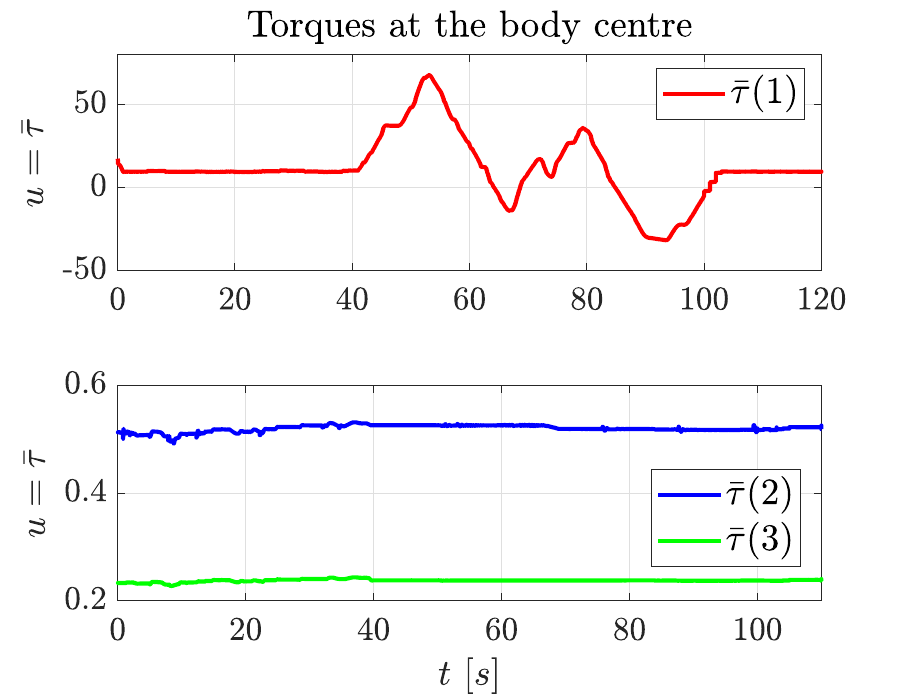}	\caption{Illustrates the control input plot; the increase in control effort is attributed to the obstacle in its path, which necessitates additional input to overcome.}
	\label{torqueG}
\end{figure}
	\begin{figure*}[htpb!]
	\centering
	\includegraphics[scale=0.125]{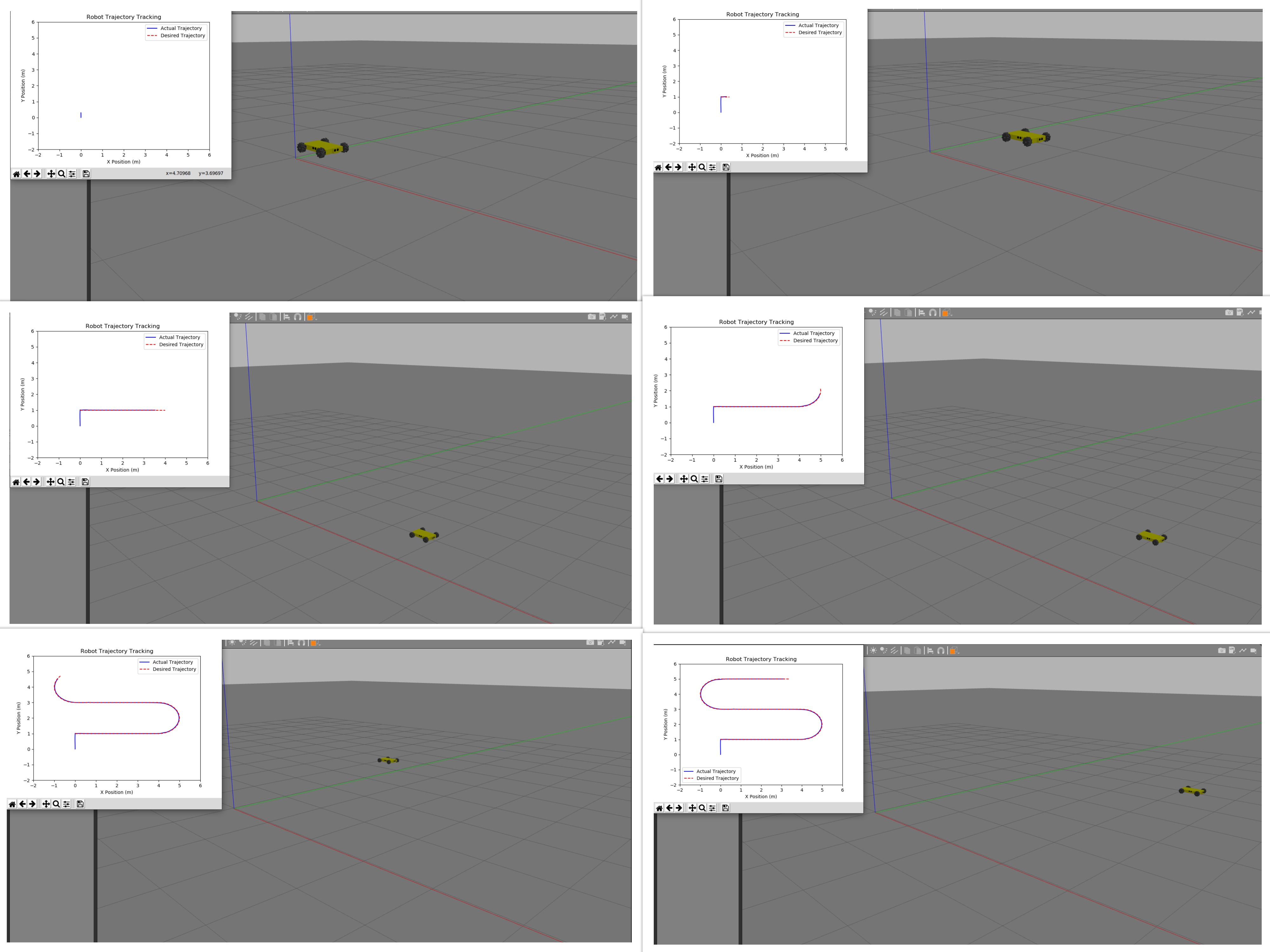}
	\caption{Few snapshots of S shaped trajectory tracking of FWMR in Gazebo-ROS considering zero initial condition for the system states; Sequence of images progresses from the top left to the bottom right, capturing key moments of the trajectory.}
	\label{Sgazebo}
\end{figure*}
\begin{figure*}[htpb]
	\centering
	\includegraphics[scale=0.125]{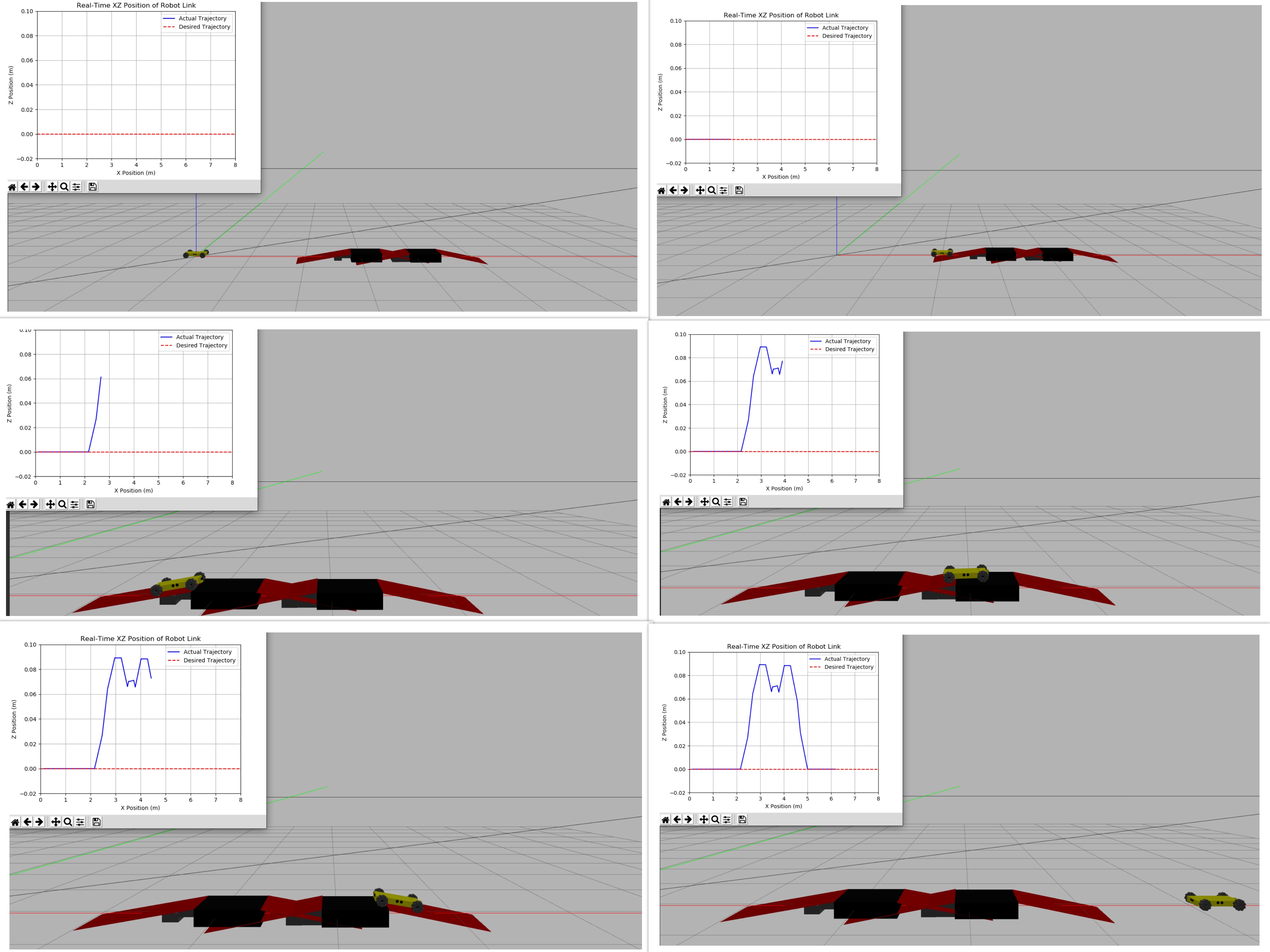}
	\caption{Few snapshots illustrating obstacle avoidance of FWMR in Gazebo-ROS where an obstacle is situated along the maneuver of the vehicle; Sequence of images progresses from the top left to the bottom right, capturing key moments of the trajectory.}
	\label{Obs}
\end{figure*}

	\section{Conclusions}\label{sect:6}
	In this work, we developed and implemented a finite-time backstepping tracking control law and validated its efficacy within a real-time Gazebo ROS simulation platform. A comprehensive proof of closed-loop stability in a finite-time setting, accompanied by explicit computations of guaranteed finite-time  was provided. We also conducted a comparative analysis between finite-time and asymptotic tracking control outcomes for FWMR model in a real-time simulation setting. Our findings illustrate the superior convergence of system states to desired values  for finite-time controller compared to the asymptotic tracking controller in real-time simulation studies. Our future endeavors are aimed to implement the proposed finite-time backstepping tracking control law in real-time FWMR hardware system, anticipating additional challenges posed by deviations from intended trajectories, control oscillations, and response delays. Addressing these challenges will be pivotal in ensuring the robustness and effectiveness of our controller in real-world applications.

	%	and the continuity also holds for $\alpha=1$. This can be directly checked by taking the derivative of $\bar{\tau}$ with respect to $\alpha$ and then equating it to the limiting value of $\alpha$.
	%\end{corollary}
	%\begin{remark} \label{rem1}
	%Justification for the choice of range of $\alpha$ is based on boundedness of the control input $\bar{\tau}$. We have derivative of $\Psi({\tilde{\eta}})$ appearing in $\bar{\tau}$,
	%\begin{align*} \label{tdpsi}
	%	\dot{\Psi}(\tilde{\eta})&=\dot{Q}^T\left(\dot{\eta}_d-\frac{K_\eta \tilde{\eta}}{\left\|{\tilde{\eta}}\right\|^{1-\alpha} }\right)+Q^T\left(\ddot{\eta}_d-\frac{K_\eta \dot{\tilde{\eta}}}{\left\|{\tilde{\eta}}\right\|^{1-\alpha}}\right)\\ \nonumber&\ \ -Q^T\left(\frac{K_\eta \tilde{\eta}\tilde{\eta}^T\dot{\tilde{\eta}}(\alpha-1)}{\left\|{\tilde{\eta}}\right\|^{3-\alpha}}\right)
	%\end{align*}
	%where, it is straightforward to see that $0\less\alpha\less1$ makes $\dot{\Psi}(\tilde{\eta})$ bounded leading to bounded feedback control torque $\bar{\tau}$ which justifies the range of $\alpha$.
	%\end{remark}	
%				\appendix
%				\section{My Appendix}
%				Appendix sections are coded under \verb+\appendix+.
%				
%				\verb+\printcredits+ command is used after appendix sections to list 
%				author credit taxonomy contribution roles tagged using \verb+\credit+ 
%				in frontmatter.
%				
%				\printcredits
                \clearpage
                %\printcredits
%                \subsection*{CRediT authorship contribution statement}
%                \textbf{Anil B:} Conceptualization, Formal Analysis, Methodology, Software,
%                Investigation, Writing – original draft, Visualization.  \textbf{Mayank Pandey:}  Conceptualization, Formal Analysis, Methodology and Software. \textbf{Sneha Gajbhiye:} Project
%                administration, Writing – review \& editing, Supervision, Funding
%                acquisition.
%                \section*{Declaration of competing interest}
%                The authors declare that they have no known competing financial interests or personal relationships that could have appeared to influence the work reported in this paper.
                \section*{Acknowledgment}
                This work is partially supported by the Ministry of Electronics and Information Technology (MeitY), Government of India, under the sponsored research project 2024-214-EE-SNG-MEITY-SP.
				\bibliographystyle{cas-model2-names}
				\bibliography{cas-refs}

				%\vskip3pt
				
%				\bio{}
%				Author biography without author photo.
%				Author biography. Author biography. Author biography.
%				Author biography. Author biography. Author biography.
%				Author biography. Author biography. Author biography.
%				Author biography. Author biography. Author biography.
%				Author biography. Author biography. Author biography.
%				Author biography. Author biography. Author biography.
%				Author biography. Author biography. Author biography.
%				Author biography. Author biography. Author biography.
%				Author biography. Author biography. Author biography.
%				\endbio
%				
%				\bio{figs/cas-pic1}
%				Author biography with author photo.
%				Author biography. Author biography. Author biography.
%				Author biography. Author biography. Author biography.
%				Author biography. Author biography. Author biography.
%				Author biography. Author biography. Author biography.
%				Author biography. Author biography. Author biography.
%				Author biography. Author biography. Author biography.
%				Author biography. Author biography. Author biography.
%				Author biography. Author biography. Author biography.
%				Author biography. Author biography. Author biography.
%				\endbio
%				
%				\bio{figs/cas-pic1}
%				Author biography with author photo.
%				Author biography. Author biography. Author biography.
%				Author biography. Author biography. Author biography.
%				Author biography. Author biography. Author biography.
%				Author biography. Author biography. Author biography.
%				\endbio
%				
			\end{document}